\documentclass[aps,pra,twocolumn,superscriptaddress,nofootinbib]{revtex4}

\usepackage{mathpazo}

\usepackage{amsmath}
\usepackage{amsfonts,amssymb,amsthm, bbm,braket}
\usepackage{graphicx}   
\usepackage{subfigure}  
\usepackage{amsbsy} 
\usepackage[bold]{hhtensor} 
\usepackage{natbib}
\usepackage{algpseudocode}
\usepackage{xcolor}

\usepackage{multirow}

\newcommand{\Cli}{\mathrm{Cl}}
\newcommand{\bbF}{\mathbb{Z}} 
\newcommand{\Sp}{\mathrm{Sp}}

\newcommand{\RR}{\mathbb{R}}
\newcommand{\error}{\epsilon}
\DeclareMathOperator{\Sym}{Sym}

\newtheorem{definition}{Definition}
\newtheorem{corollary}{Corollary}
\newtheorem{theorem}{Theorem}
\newtheorem{lemma}{Lemma}
\newtheorem{proposition}{Proposition}

\definecolor{myblue}{RGB}{51,51,178}
\definecolor{mygreen}{RGB}{0,128,0}
\definecolor{myred}{RGB}{189,26,26}

\newcommand{\vc}[1]{#1}
\newcommand{\mt}[1]{#1}

\newcommand{\tr}{\mathrm{tr}}
\newcommand{\sym}{\mathrm{Sym}}

\newcommand{\CC}{\mathbb{C}}
\newcommand{\ZZ}{\mathbb{Z}}

\usepackage[breaklinks=true]{hyperref}

\hypersetup{
  colorlinks   = true, 
  urlcolor     = blue, 
  linkcolor    = blue, 
  citecolor   = green 
}

\begin{document}

\title{Low rank matrix recovery from Clifford orbits}

\author{Richard Kueng}
\author{Huangjun Zhu}

\affiliation{Institute for Theoretical Physics, University of Cologne, Germany}

\author{David Gross}
\affiliation{Institute for Theoretical Physics, University of Cologne, Germany}
\affiliation{Centre for Engineered Quantum Systems, School of Physics,
The University of Sydney, Sydney, NSW 2006, Australia}

\date{\today}


\begin{abstract}
We prove that low-rank matrices can be recovered efficiently from a small number of measurements that are sampled from orbits of a certain matrix group.
As a special case, our theory makes statements about the \emph{phase retrieval} problem.
Here, the task is to recover a vector given only the \emph{amplitudes} of its inner product with a small number of vectors from an orbit.
Variants of the group in question have appeared under different names in many areas of mathematics.
In coding theory and quantum information, it is the \emph{complex Clifford group}; in time-frequency analysis the \emph{oscillator group}; and in mathematical physics the \emph{metaplectic group}.
It affords one particularly small and highly structured orbit that includes and generalizes the discrete Fourier basis:
While the Fourier vectors have coefficients of constant modulus and phases that depend linearly on their index, the vectors in said orbit have phases with a quadratic dependence.
In quantum information, the orbit is used extensively and is known as the set of \emph{stabilizer states}.
We argue that due to their rich geometric structure and their near-optimal recovery properties, stabilizer states form an ideal model for structured measurements for phase retrieval.
Our results hold for $m \geq C  \kappa_r r d \log (d)$ measurements, where the  oversampling factor $\kappa_r$ varies between $\kappa_r=1$ and $\kappa_r=r^2$ depending on the orbit.
The reconstruction is stable towards both additive noise and deviations from the assumption of low rank.
If the matrices of interest are in addition positive semidefinite, reconstruction may be performed by a simple constrained least squares regression.
Our proof methods could be adapted to cover orbits of other groups.
\end{abstract}


\maketitle

\section{Motivation}

\subsection{Phase retrieval}

Starting point of this paper is the \emph{phase retrieval problem}
\cite{walther_question_1963}.
The problem is to reconstruct an unknown vector $\vc{x} \in \mathbb{C}^d$ from measurements of the form
\begin{equation}
	y_k = \left| \langle \vc{a}_k, \vc{x} \rangle \right|^2 + \epsilon_k \quad 1 \leq k \leq m.	\label{eq:phaseless_measurements}
\end{equation}
Here, the $\vc{a}_1,\ldots,\vc{a}_m \in \mathbb{C}^d$ model linear measurements and the $\epsilon_k$'s additive noise. 
The phase retrieval problem occurs in many areas of science, for example in X-ray crystallography \cite{millane_phase_1990}, astronomy \cite{fienup_phase_1987,fienup_hubble_1993} and diffraction imaging \cite{chapman_femtosecond_2006, pfeiffer_phase_2006}, as well as pure state quantum estimation theory \cite{flammia_minimal_2005, peres_quantum_2006, heinosaari_quantum_2013, baldwin_strictly_2016,carmeli_efficient_2016}.

While a  number of (often heuristic) algorithms for solving the inverse problem (\ref{eq:phaseless_measurements}) have long been used \cite{fienup_phase_1982}, a rigorous analysis is quite involved. 
Even in the absence of noise ($\epsilon_k=0$), it is not obvious how many measurments are necessary, which measurements can be employed, and whether the vector $\vc{x}$ can be recovered in a numerically stable and computationally efficient way.

Recently, techniques from convex optimization theory have been used with great success to analyze the phase retrieval problem. 
This ansatz is known as \emph{PhaseLift} \cite{candes_phase_2013, candes_exact_2013}.
Beyond suggesting an algorithm for reconstruction, it brings powerful methods---e.g.\ convex duality theory---into the fold. 
There is now a fast-growing body of literature (including \cite{candes_solving_2013, gross_partial_2014, alexeev_phase_2014, candes_coded_2015, gross_improved_2015, kueng_orthonormal_2015, kueng_low_2015, salanevich_polarization_2015, tropp_convex_2015, krahmer_phase_2016, kabanava_stable_2015}) using these tools to establish recovery guarantees for phase retrieval for a variety of measurement models.

The present paper is also based on PhaseLift, and we will review the technique in Section~\ref{sec:convex} below.
However, our main focus lies on identifying a new class of measurement vectors ${a}_1, \dots, {a}_m$ for which the phase retrieval problem can be proven to be well-posed. 
Consequently, we start the introduction by discussing these measurement models.

\subsection{Measurement models}

While some deterministic sets of measurements vectors for the phase retrieval problem have recently been described \cite{kech_explicit_2015,bodmann_algorithms_2016,carmeli_efficient_2016, pohl_phaseless_2014}, 
most constructions are randomized.
The typical result states that if the $m$ vectors ${a}_k$ are drawn independently according to some distribution in $\mathbb{C}^d$, the inverse problem associated with \eqref{eq:phaseless_measurements} is well-posed with overwhelming probability.
An incomplete and \emph{ad hoc} classification of known examples might look as follows:

\emph{Gaussian or Haar distributions}.---The strongest, easiest to prove, and earliest examples used Gaussian random vectors, or vectors drawn uniformly from the unit-sphere (\emph{Haar} distribution) \cite{candes_exact_2013,candes_solving_2013}.
While powerful, this ensemble is rarely suitable in practical applications and gives no indication as to which particular properties are necessary for phase retrieval.
These two deficits are addressed by the next two categories.

\emph{Ensembles modeled after particular applications}.---Measurements modeling practical applications have been analyzed. 
An early example is given by the works on \emph{coded diffraction patterns} that are motivated by problems arising in diffraction imaging \cite{candes_coded_2015,gross_improved_2015}.
While highly relevant, the arguments tend to be very specific to the particular use case.

\emph{Designs}.---In contrast, one can ask for weak abstract properties of ensembles that are sufficient for phase retrieval. 
Since the proofs establishing recovery guarantees typically require information on higher moments $\mathbb{E} \left[{a}^{\otimes t} ({a}^*)^{\otimes t}\right]$ of the random measurement vectors, Ref.~\cite{gross_partial_2014} analyzed the suitability of \emph{complex projective $t$-designs} \cite{delsarte_spherical_1977,renes_symmetric_2004} for phase retrieval, see also \cite{kueng_spherical_2015, ehler_phase_2015}. 
These are ensembles that reproduce the first $2t$-th moments of the uniform distribution on the sphere.
This program has been successful in the sense that excellent recovery guarantees for designs of degree $t\geq 4$ have been established \cite{kueng_low_2015}.
At the same time, it has not yet lived up to some early expectations, because constructions known for infinite families of $4$-designs \cite{ambainis_quantum_2007,brandao_local_2012} are arguably significantly less explicit and ``well-structured'' as is the case for $t=3$ \cite{zhu_multiqubit_2015, webb_clifford_2015, kueng_qubit_2015}, or lower \cite{klappenecker_mutually_2005,bodmann_achieving_2015}.
Amending this situation was one of the motivations for the present work.

Here, we pursue a different route and establish representation-theoretic techniques for deriving recovery guarantees for measurement vectors sampled uniformly from orbits of matrix groups.
We pay particular attention to the \emph{complex Clifford group} and its orbit of \emph{stabilizer states}, introduced below.

Stabilizer states have appeared in several areas of science.
They are particularly central to quantum information theory. 
Therefore, this (and closely related \cite{povm_paper}) results might find direct practical applications e.g.\ in quantum state estimation \cite{perina_quantum_2012,banaszek_focus_2013}.

Beyond that, however, we want to put forth the argument that stabilizer states provide an ideal model for phase retrieval measurements, due to their rich geometric structure and the near-optimal recovery guarantees that can be proved.

The remaining paragraphs of this section are somewhat subjective and speculative.
Readers more interested in mathematical meat than meta-mathematical chatter should skip ahead.

To explain what we have in mind, we recall the situation for related inverse problems that can be tackled using convex optimization theory \cite{foucart_mathematical_2013}: \emph{compressed sensing} for sparse vectors, and \emph{low-rank matrix recovery}.

Also in compressed sensing, first results pertained to Gaussian measurements \cite{donoho_compressed_2006, candes_near_2006, foucart_mathematical_2013}. 
Attention then quickly shifted to ``more structured'' models, with the most natural one being measurements sampled from the Fourier basis \cite{candes_stable_2006,candes_robust_2006}.
Beyond its practical relevance, the high degree of a geometric and algebraic structure connected to Fourier vectors makes their study particularly fruitful. 
The absolutely tight bounds in \cite{tao2003uncertainty} serve as one example.

The story is similar in low-rank matrix recovery. 
Initial results are for Gaussian measurements \cite{recht_guaranteed_2010}, with follow-up work concentrating on the practically relevant measurement model of random matrix elements \cite{candes_exact_2009,candes_power_2010,keshavan_matrix_2010}.
A measurement model whose ``high degree of structure'' allows for a particularly simple analysis is given by the \emph{Pauli basis} \cite{gross_quantum_2010,gross_recovering_2011, liu_universal_2011,FlamGLE12}. 
For example, the fact that it constitutes a \emph{unitary operator basis} yields short proofs and tight bounds on the sampling rate.
What is more, its rich algebraic structure has been a crucial ingredient to a matching converse bound \cite[Theorem~15]{gross_recovering_2011}.

Arguably, a measurement model for phase retrieval, that would take up a role analogous to the ones plaid by the Fourier and Pauli bases, has not yet emerged.
A main reason may be the lack of obvious candidates: 
There are just not that many infinite families of high-dimensional vector configurations that have been widely studied.
To the present authors, the set of stabilizer states constitutes a worthy candidate for that role.

\section{Stabilizer states and the Clifford group}

Here, we introduce the concepts of \emph{stabilizer states} and the (complex) \emph{Clifford group} from various points of view.
These, and related notions, have been discovered several times in different branches of science and mathematics, including quantum information theory \cite{gottesman_stabilizer_1997}, coding theory \cite{macwilliams1977theory, nebe_self_2006}, time-frequency analysis \cite{grochenig2013foundations,pfander-chapter}, as well as mathematical physics and functional analysis \cite{v1931eindeutigkeit,mackey1955theory,Foll89book}.

Owing to the upbringing of the authors, we initially adopt the vantage point of quantum information, where these concepts go back to Ref.~\cite{gottesman_stabilizer_1997}.  
The textbook \cite{nielsen_quantum_2010} treats them extensively.
A point of view focused on connections to symplectic geometry is given in Refs.~\cite{gross_hudson_2006,kueng_qubit_2015}.
Here, we mainly summarize results from these sources.

In Section~\ref{sec:related_groups}, we comment on related gropus and give some pointers to the literature of other fields (though we are by no means experts in this regard).

\subsection{Stabilizer states: Elementary approach}
\label{sec:stabsElementary}

In this section, we provide a concrete basis representation of stabilizer states. 
While very transparent, it turns out that calculations are best done using an indirect representation in terms of \emph{stabilizer groups} or certain structures from symplectic geometry.
This point of view is described in the following sections. 

We first recall the definition of the Fourier  basis 
associated with the discrete vector space $\ZZ_2^{n}$.
To this end, we label the standard basis $\{\vc{e}_x\}_x$ of $\mathbb{C}^{2^n}$ by vectors $x\in\ZZ_2^n$.
A \emph{Fourier vector} $\vc{f}_l\in\CC^{2^n}$ depends on a $\ZZ_2$-linear form $l:\ZZ_2^n\to \ZZ_2$ and has expansion coefficients given by
\begin{equation}\label{eq:fourierbasis}
	\left(\vc{f}_{l}\right)_x
	=
	2^{-n/2}\,(-1)^{ l(x) }.
\end{equation}

Essentially, stabilizer states are obtained by generalizing (\ref{eq:fourierbasis}) to allow for a \emph{quadratic} dependence of the phase on the label $x$.
Indeed, the simplest type of stabilizer state $\psi_q$ is defined by a
\emph{quadratic form} $q$ on $\ZZ_2^n$ and has expansion coefficients
\begin{equation}\label{eq:stabsimple}
	(\psi_q)_x = 2^{-n/2}\,(-1)^{q(x)}.
\end{equation}
The most general form goes beyond Eq.~(\ref{eq:stabsimple}) in two ways: (i) The
non-zero coefficients can be restricted to an affine subset of
$\ZZ_2^n$, and (ii) certain complex phase factors are also allowed.
To be precise \cite{hostens2005stabilizer,gross2008lu}: 

\begin{definition}
	A \emph{stabilizer state} $\psi \in \CC^{2^n}$ is defined by the following data:
	\begin{enumerate}
		\item
		An affine subset $A\subset \ZZ_2^n$,
		\item
		a quadratic form $q:A\to\ZZ_2$,
		\item
		a linear form $l:A \to \ZZ_2$.
	\end{enumerate}
	Its components are
	\begin{equation*}
		\psi_x = 
		\left\{
			\begin{array}{ll}
				|A|^{-1/2}\,i^{l(x)} (-1)^{q(x)} & x \in A \\
				0 & x \not\in A.
			\end{array}
		\right.
	\end{equation*}
\end{definition}

More explicitly, recall that the standard inner product gives rise to a one-one correspondence between linear forms $l_y$ on $\ZZ_2^n$ and elements $y$ of $\ZZ_2^n$ via
\begin{equation*}
	l_y(x) := \langle y , x \rangle = \sum_{i=1}^n y_i x_i \mod 2.
\end{equation*}
A quadratic form $q$ on $\ZZ_2^n$ is a function that can be written as
\begin{equation}\label{eq:quadratic}
	q(x) = \sum_{i\leq j} q_{i,j}\,x_i x_j\mod 2,
\end{equation}
for some upper triangular matrix $q_{i,j}$.
Because in characteristic $2$, it holds that $x_i^2 = x_i$, Eq.~(\ref{eq:quadratic}) inlcudes linear forms as a special case:
\begin{equation*}
	\langle x, y \rangle = \sum_{i\leq j}
	\operatorname{diag}(y)_{i,j}\,x_i x_j,
\end{equation*}
where $\operatorname{diag}(y)$ is the matrix with the vector $y$ on its main diagonal.
In particular, the Fourier basis is included in the set of stabilizer states.

Many properties of stabilizer states are known. 
E.g.\ they can be partitioned into disjoint sets of ortho-normal bases; and the inner product between two stabilizer states depends essentially only on the basis they belong to.
All these properties are hard to see from their basis expansion. 
The formalism of \emph{stabilizer groups}, introduced next, is less explicit, but makes such computations much easier.

\subsection{Stabilizer states from stabilizer groups}
\label{sec:stabsGeometric}

Stabilizer states can be defined implicitly as the common eigenvectors of maximal sets of commuting Pauli operators.

To make this precise, define the \emph{Pauli operators} on $\CC^2$ as
\begin{align}
	\sigma_{(0,0)} =&  
	\left(
		\begin{array}{cc}
			1 & 0 \\
			0 & 1 
		\end{array}
	\right), 
&
	\sigma_{(0,1)} =  &
	\left(
		\begin{array}{cc}
			0 & 1 \\
			1 & 0 
		\end{array}
	\right), \label{eq:singlequbitpaulis}\\
	\sigma_{(1,0)} =&  
	\left(
		\begin{array}{cc}
			1 & 0 \\
			0 & -1 
		\end{array}
	\right),
&
	\sigma_{(1,1)} =  &
	\left(
		\begin{array}{cc}
			0 & -i \\
			i &0 
		\end{array}
	\right). \nonumber
\end{align}

A Pauli operator $\mt{W}_a$ on 
\begin{equation*}
	\underbrace{\CC^2\otimes\dots\otimes \CC^2}_{n \text{ factors}}
	\simeq
	\CC^{2^n}
\end{equation*}
is defined as the tensor product of $n$ such matrices:
\begin{equation}
	\mt{W}_a := \sigma_{(a_1,a_2)} \otimes \dots \otimes
	\sigma_{(a_{2n-1},a_{2n})}. \label{eq:pauli_label}
\end{equation}
Clearly, the index $a$ takes values in the $\ZZ_2$-vector space $\ZZ_2^{2n}$.

We denote the set of all Pauli operators on $\CC^{2^n}$ by
\begin{equation*}
	\bar{\mathcal{P}}_n = \{ \mt{W}_a \,|\, a\in\ZZ_2^{2n} \}.
\end{equation*}
The \emph{Pauli group} $\mathcal{P}_n$ is the group generated by all the Pauli operators in $\bar{\mathcal{P}}_n$.
It turns out that the group consists of the Pauli operators multiplied by phase factors that are powers of the 
imaginary unit  $i$:
\begin{equation*}
\mathcal{P}_n = 
\langle \bar{\mathcal{P}}_n \rangle
= \{ i^j \mt{W}_a \,|\, a\in\bbF_2^{2n}, j \in \ZZ_4 \}.
\end{equation*}

\begin{definition}\label{def:stabilizergroup}
	A \emph{stabilizer group} is a subgroup $S \subset \mathcal{P}_n$ of the Pauli group such that
	\begin{enumerate}
		\item
		$S$ is abelian,
		\item
		$S$ does not contain $-\mathbb{1}$,
		\item
		$S$ has cardinality\footnote{
  There is a theory that generalizes stabilizer states to so-called \emph{stabilizer codes} \cite{gottesman_stabilizer_1997,nielsen_quantum_2010} and, for that purpose, drops the restriction on the cardinality of $S$.
  However, we do not require these concepts in the present paper and thus stick to the more restrictive definition.
		}
		$|S|=2^n$.
	\end{enumerate}
\end{definition}

Because a stabilizer group $S$ is abelian, there is an eigenbasis common to all $\mt{W} \in S$.
It turns out that this basis is unique (up to phase factors) and given by stabilizer states.
Every stabilizer state arises this way. 

In fact, it suffices to consider the joint $(+1)$-eigenspace:

\begin{proposition}
	There is a one-one correspondence between stabilizer states and stabilizer groups.

	Given a stabilizer state $\psi\in\CC^{2^n}$, the associated stabilizer group is
	\begin{equation*}
		S = \{ \mt{W} \subset \mathcal{P}_n, \,|\, \mt{W} \psi = \psi\}.
	\end{equation*}

	Given a stabilizer group $S\subset\mathcal{P}_n$, a projection onto the associated stabilizer state is
	\begin{equation*}
		\psi \psi^*
		=
		2^{-n} \sum_{\mt{W} \in S} \mt{W}.
	\end{equation*}
\end{proposition}

In quantum information theory, stabilizer states are usually introduced this way, and not using the basis representation we gave in the previous section. 
This explains the name.

\subsection{The symplectic connection}

There is a more abstract way to look at stabilizer states in terms of certain objects in discrete symplectic vector spaces (c.f.\ e.g.\ \cite{gross_hudson_2006, kueng_qubit_2015} and references therein).
We will briefly introduce this connection next.

Symplectic structures appear in the composition law and the commutation relation of Pauli operators.
To explain this, let $\mt{J}$ be the $2n\times 2n$ block-diagonal matrix with $n$ blocks of $\left(\begin{smallmatrix} 0 &1\\ 1 &0 \end{smallmatrix}\right)$ on the diagonal. 
Then
\begin{equation*}
	[\vc{a},\vc{b}]=\vc{a}^{\mathrm{T}} \mt{J} \vc{b}
\end{equation*}
defines a \emph{symplectic form} on $\ZZ_2^{2n}$.
One can then verify that the communtation relation 
\begin{eqnarray}
	\mt{W}_a \mt{W}_b &=&(-1)^{[ a,b]}\mt{W}_b \mt{W}_a \label{eq:pauli_commutation_relation}
\end{eqnarray}
holds.
With a slight abuse of notation, the group law of the Pauli group can be written as
\begin{eqnarray}
	\mt{W}_a \mt{W}_b &=& i^{[a,b]} \mt{W}_{a+b}, \label{eq:pauli_group_law} 
\end{eqnarray}
where the arithemtic in the exponent of $i$ is to be performed modulo $4$ (as opposed to in $\ZZ_2$).

We can use these relations to analyze the structure of stabilizer groups.
Let $S\subset\mathcal{P}_n$ be a stabilizer group.
It is of the form
\begin{equation}\label{eq:stabilizergroup}
	S = \{ (-1)^{s(a)} \mt{W}_a \,|\, a \in M \}
\end{equation}
for some set $M\subset \ZZ_2^{2n}$ and some function $s: M \to \ZZ_2$.
Pauli operators with complex coefficients $\pm i\,\mt{W}_a$ cannot occur inside a stabilizer group, for else their square
\begin{equation*}
	\big( \pm i\,\mt{W}_a \big)^2 = - \mathbb{1}
\end{equation*}
would also be an element of the group, contrary to Definition~\ref{def:stabilizergroup}.

Because $S$ is a group and because (\ref{eq:pauli_group_law}) states that composition of Pauli operators $W_a W_b$ corresponds to addition $a+b$ of their indices, it follows that $M$ is closed under addition and hence a subspace of $\ZZ_2^{2n}$. 
The fact that $S$ is abelian and the commutation relation (\ref{eq:pauli_commutation_relation}) together imply that the symplectic form $[\cdot, \cdot]$ vanishes on $M$. 
Such spaces are called \emph{isotropic} in symplectic geometry.
Finally, the fact that $|S|=2^n$ means that $M$ has dimension $n$ as a subspace of $\ZZ_2^{2n}$.
As isotropic subspaces can have at most half the dimension of the ambient vector space, this means that $M$ is a \emph{maximal isotropic subspace}, or a \emph{Lagrangian subspace}.

Next, we turn to the phase function $s: M \to \ZZ_2$ defined in (\ref{eq:stabilizergroup}).
Its value can be chosen freely on a basis $\{b_1, \dots, b_n\}$ of $M$.
This choice gives rise to a generating set 
\begin{equation*}
	S = 
	\langle
		\{
			(-1)^{s(b_1)} W_{b_1},
			\dots,
			(-1)^{s(b_n)} W_{b_n}
		\}
	\rangle,
\end{equation*}
which extends $s$ uniquely to all of $M$.
In this way, one obtains $2^n$ different stabilizer groups for any given $M$ and one can show that this set does not depend on the choice of basis for $M$.

The $2^n=d$ stabilizer states associated to any given Lagrangian subspace $M$ turn out to form an orthonormal basis for $\CC^d$.
Thus, the set of stabilizer states can be partitioned into disjoint orthonormal bases.

To summarize:

\begin{proposition}
	A stabilizer group $S\subset\mathcal{P}_n$ can be specified by the following data:
	\begin{enumerate}
		\item
		A Lagrangian subspace $M\subset\ZZ_2^{2n}$,
		\item
		a phase function $s: M \to \ZZ_2$, which can be freely chosen on a basis $\{b_1, \dots, b_n\}$ of $M$.
	\end{enumerate}
	The stabilizer group $S$ is then generated by the $n$ Pauli operators $(-1)^{s(b_k)} W_{b_k}$ for $k=1, \dots, n$.
\end{proposition}

\subsection{Symmetries: The Clifford group}

Most important for our analysis below is the fact that the set of stabilizer states affords a large, transitive symmetry group. 
To introduce it, we define the \emph{Clifford group} as follows\footnote{
We remark that the term \emph{Clifford group} sometimes refers to a minor variant of the group introduced here. 
Indeed, note that if $\mt{U} \in\Cli_n$, then so is $e^{i\phi} \mt{U}$ for every phase $\phi$. 
For our purposes, these phase factors are unimportant.
But often, it is desirable to work with a version of the Clifford group that includes as few phases as possible in the sense that its intersection with the center $Z=\{ e^{i\phi} \mathbb{1} \}_{\phi\in\mathbb{R}}$ of $U(2^n)$ is the smallest.
One can find explicit generators for a group $\Cli_n'$ which is identical to $\Cli_n$ up to phases in that $\Cli_n'/Z = \Cli_n/Z$ and such that $\Cli_n'\cap Z = \{ i, -1, 1, -i \}$, which is minimal \cite{other_paper}.
We also remark that the term \emph{Clifford group} is sometimes used to refer to the cover group of the orthogonal group that is given by the invertible elements inside a \emph{Clifford algebra}. 
Despite this unfortunate coincidence in names, there seems to be no connection between this group and the one used here.
}.

\begin{definition}
	The \emph{Clifford group} $\Cli_n$ is the  \emph{normalizer} of $\mathcal{P}_n$ inside of $U(2^n)$. 
\end{definition}

In other words, $\Cli_n$ is the set of unitaries $\mt{U}$ such that, 
for all Pauli operators $\mt{W}_a$, it holds that
\begin{equation*}
	\mt{U} \mt{W}_a \mt{U}^\dagger \in \mathcal{P}_n.
\end{equation*}

Because the Clifford group maps elements of the Pauli group to elements of the Pauli group under conjugation, it also maps stabilizer groups to stabilizer groups.
By the preceding sections, this means that the Clifford group maps stabilizer states onto stabilizer states.
That action is known to be transitive---i.e.\ the set of stabilizer states forms an orbit under the Clifford group.

From Equation~(\ref{eq:pauli_commutation_relation}), it follows that for any pair of Pauli operators $\mt{W}_a, \mt{W}_b$ it holds that
\begin{equation*}
	\mt{W}_a \mt{W}_b \mt{W}_a^\dagger = \pm \mt{W}_b.
\end{equation*}
Thus the Pauli group forms a subgroup of the Clifford group $\mathcal{P}_n \subset \Cli_n$.

More interesting is the quotient of the Clifford group up to phases and the Pauli group.
To explain it, note that since an action by conjugation preserves group laws and since the group law (\ref{eq:pauli_group_law}) of the Pauli group is tied to both the linear structure and the symplectic form on $\ZZ_2^{2n}$, it seems plausible that discrete symplectic groups might play a role.
This is indeed true.
Let $\Sp(2n,\ZZ_2)$ be the \emph{symplectic group} composed of all $2n\times 2n$ matrices  $\mt{F}$ over $\ZZ_2$ that satisfy the relation
\begin{equation}
\mt{F}\mt{J}\mt{F}^T=\mt{J}.
\end{equation}
Then we have:

\begin{proposition}
  For every $\mt{U}\in \Cli_n$, there is a unique symplectic matrix $\mt{F}\in\Sp(2n,\ZZ_2)$ such that
  \begin{equation}\label{eqn:symplecticaction}
  \mt{U} \mt{W}_a \mt{U}^\dagger = (-1)^{f(a)} \mt{W}_{Fa}\qquad \forall a\in \bbF_2^{2n},
  \end{equation}
  where  $f$ is a  function from $\bbF_2^{2n}$ to $\bbF_2$. Conversely, for each symplectic matrix $\mt{F} \in\Sp(2n,\ZZ_2)$ there exists a $\mt{U}\in \Cli_n$ and a suitable function $f$ such that the above equation is satisfied.  
\end{proposition}

Note that  Clifford unitaries of the form $e^{i\phi} \mt{U} \mt{W}_a$ for $\phi\in\mathbb{R}$ and $a\in \bbF_2^{2n}$ induce the same symplectic transformation. 
In fact, the quotient of $\Cli_n$ up to the Pauli group and phase factors is isomorphic to $\Sp(2n,\ZZ_2)$.

Thus, not only does the set of stabilizer states afford a transitive symmetry group, it is also true that the group has a geometric interpretation, in terms of symmetries of symplectic vector spaces.
The geometric description of the Clifford group enables some explicit calculations that are crucial ingredients to our main result: Theorem~\ref{thm:main}. 
Indeed, this statement depends on an analysis of the representation theory of tensor powers of the Clifford group \cite{other_paper}, which in turn relies on counting arguments involving orbits of tuples of vectors $\langle v_1, \dots, v_k \rangle \in (\ZZ_2^{2n})^{\times k}$  under the action of $\Sp(2n,\ZZ_2^{2n})$ \cite{zhu_multiqubit_2015}.

We feel that this lends credence to our earlier claim that their ``rich geometric structure'' makes stabilizer states into an ideal model measurement ensemble.

\subsection{Related groups and uses in other fields}
\label{sec:related_groups}

The Pauli group discussed in this paper is strongly related to the \emph{Heisenberg groups} and their \emph{Weyl representations}.
These appear in a number of fields.
The literature on this subject is vast and seems, unfortunately, to be quite disconnected.
We will not describe a unifying theory here, but merely mention some examples and how they relate to this work.

A variant that is of importance in in quantum mechanics, functional analysis \cite{Foll89book}, and time-frequency analysis \cite{grochenig2013foundations} involves operators acting on $\mathrm{L}^2(\mathbb{R})$, the set of square-integrable functions on the real line.
One way to approach it is to start with operators $P,Q$ that satisfy the \emph{canonical commutation relations}
\begin{equation*}
	[Q,P]=i\,\mathbb{1}.
\end{equation*}
This relation makes the linear space spanned by $P,Q$ into a Lie algebra. 
The elements of the associated Lie group are sometimes referred to as \emph{Weyl operators} and parameterized as
\begin{equation*}
	W_{p,q} = e^{-i\frac12 pq} e^{ip Q} e^{i q P}
\end{equation*}
for $(p,q)\in\RR^2$.
The group law can be verified to be
\begin{equation*}
	W_{a} W_{b} = e^{i \frac12 [a,b]} w(a+b),
\end{equation*}
for $a,b\in\RR^2$ and $[\cdot,\cdot]$ the standard symplectic form on $\RR^2$.
This is clearly analogous to the corresponding law (\ref{eq:pauli_group_law}) for the Pauli group.
The normalizer of the Weyl operators -- i.e.\ the analogue of the Clifford group -- is often called the \emph{metaplectic group} and is related to the symplectic group $\Sp(\RR^2)$ \cite{Foll89book}.
In time-frequency analysis, Weyl operators are usually known as \emph{time-frequency shifts} and the metaplectic group is sometimes referred to as the \emph{oscillator group} \cite{grochenig2013foundations}.
The relatives of stabilizer states are complex Gaussian vectors that play an important role e.g.\ in quantum optics \cite{walls_quantum_2008,hudson_wigner_1974}.

The Weyl operators act on functions $\psi\in\mathrm{L}^2(\mathbb{R})$ as
\begin{equation}\label{eqn:weylaction}
	\big(W_{p,q}\psi\big)(x)
	=
	e^{-i\frac12 pq} 
	e^{ipx}
	\psi(x-q).
\end{equation}
This formula suggests a natural way to define discretized versions of the Weyl operators: 
Just re-interpret the numbers $p,q,x$ as elements of $\ZZ_d$ for some natural number $d$ to obtain versions of $\mt{W}_{p,q}$ acting on $\mathrm{L}^2(\ZZ_d) \simeq \CC^d$.
It turns out 
\cite{debeaudrap_linearized_2011, appleby_symmetric_2005}
that the theory becomes slightly cleaner if one also  changes the phase factor in (\ref{eqn:weylaction}) as
\begin{equation*}
	e^{-i\frac12 pq} \mapsto \tau^{pq}, 
	\qquad
	\tau := \mathrm{e}^{\pi i (d^2+1)/d}=(-1)^d \mathrm{e}^{\pi i /d}.
\end{equation*}
This procedure does, in fact, define unitary operators on $\CC^d$ (and recovers the Pauli operators $\bar{\mathcal{P}}_1$ when setting $d=2$).

These discrete Weyl-Heisenberg operators are usually introduced from a different point of view.
To state it, let $\{\vc{e}_1, \dots, \vc{e}_d\}$ be the standard basis of $\CC^d$ and define
\begin{eqnarray*}
	\mt{X}: \vc{e}_{k} &\mapsto& \vc{e}_{k+1},  \\
	\mt{Z}: \vc{e}_{k} &\mapsto& \omega_d^k\,\vc{e}_{k},
\end{eqnarray*}
where $\omega_d=e^{i 2\pi / d}=\tau^2$ is a $d$-th root of unity.
Then the discrete Weyl-Heisenberg operators can also be written as
\begin{equation}\label{eq:generalized}
	\mt{W}_{p,q} = \tau^{pq} \mt{X}^q \mt{Z}^p.
\end{equation}
These operators are also known as \emph{generalized Pauli operators}.

Once again, one can define an associated Clifford group as the normalizer of these $W_{p,q}$'s and introduce stabilizer groups and states that are compatible with this structure.

Further related groups and applications come from coding theory. 
With every binary code, one can associate the \emph{weight enumerator polymonial} whose coefficients encode the number of codewords of a given weight.
For certain classes of self-dual codes, one can fairly easily see that these polynomials are invariant under an associated symmetry group \cite{macwilliams1977theory}.
The complex Clifford group used in this paper appears here, but also a real-valued variant, which is more commonly studied in this context \cite{nebe_self_2006}.
In our language, the real Clifford group arises as the normalizer of the group generated by the real Pauli operators.
This real group arises in many other contexts, e.g.\ as the symmetry group of the \emph{Barnes-Wall lattice}.
A good starting point to the literature covering this approach is the book \cite{nebe_self_2006} (in particular the Background section of their Chapter~6).

Given all these similarities, it is natural to ask whether the low-rank recovery results of the present paper can be adapted to these more general Clifford groups and stabilizer states.

For the set of Weyl operators that appear in (\ref{eq:generalized}), it seems clear that if similar results could be established, new techniques would have to be developed for this purpose.
Indeed, our proof relies crucially on the representation theory of the fourth tensor power of the Clifford group \cite{other_paper,helsen_representations_2016} to derive bounds on the $8$th moments of random vectors sampled from orbits.
It is known \cite{zhu_multiqubit_2015, webb_clifford_2015} that the representation theory of the particular Clifford group studied here behaves differently from its cousins already for the third tensor power.
What is more, there is a precise sense in which the moment bounds are worse for these more general stabilizer states: they fail to form a \emph{complex projective $3$-design} \cite{kueng_qubit_2015}. 
We refer to Section~\ref{sec:moments} below for an introduction of the design concept.

Deciding whether our results can be translated to the more general case despite these obstacles is an interesting open problem.

The situation might be better for the real Clifford group. 
Here, analogous representation-theoretic results to those proven in \cite{other_paper} in the complex case had been known for some time \cite{Runge96,nebe_self_2006}.

Deciding whether our results can be translated to the more general case remains an interesting open problem.

\section{Convex low-rank recovery and PhaseLift} \label{sec:convex}

Here, we briefly review some basic facts from the theory of convex low-rank recovery. 
We refer to Ref.~\cite{foucart_mathematical_2013} for a more thorough introduction.

Building on ideas from \emph{compressed sensing} \cite{foucart_mathematical_2013}, low rank matrix recovery aims at reconstructing unknown $d \times d$ matrices $\mt{X}$ of rank $r$ from few noisy linear measurements of the form
\begin{equation}
y_k = \tr \left( \mt{A}_k \mt{X} \right) + \epsilon_k \quad 1 \leq k \leq m. \label{eq:measurements}
\end{equation}
For simplicity, the present paper restricts attention to hermitian matrices $\mt{X}, \mt{A}_k \in H_d$ -- though we expect that standard constructions can be used to lift this assumption, see e.g.\  \cite{gross_recovering_2011}.

The measurement model \eqref{eq:measurements} can be written more succinctly as
\begin{equation*}
	\vc{y} = \mathcal{A}(\mt{X}) + \error,
\end{equation*}
where $\vc{y} = (y_1,\ldots,y_m)^T \in \mathbb{R}^m$ 
represents the measurement data, $\error = (\epsilon_1,\ldots,\epsilon_m )^T$ denotes additive noise corruption and $\mathcal{A} : H_d \to \mathbb{R}^m$ is the measurement operator
\begin{equation*}
\mathcal{A}(\mt{Z}) = \sum_{k=1}^m \mathrm{tr} \left( \mt{A}_k \mt{Z} \right) \vc{e}_k,
\end{equation*}
with $\vc{e}_1,\ldots,\vc{e}_m$ being the standard basis of $\mathbb{R}^m$.

For many measurement models, it has been proven that low-rank matrices can be recovered efficiently using a constrained nuclear norm minimization:
\begin{align}
  \underset{\mt{Z} \in H_d}{\textrm{minimize}} & \quad \left\| \mt{Z} \right\|_1 \label{eq:phaselift} \\
  \textrm{subject to} & \quad \left\| \mathcal{A}(\mt{Z}) - \vc{y} \right\|_{\ell_q} \leq \eta, \nonumber \quad 1 \leq q \leq \infty.
\end{align} 
Here, $\eta \geq \| \error \|_{\ell_q}$ is an upper bound on the noise corruption in \eqref{eq:measurements} and the nuclear norm $\| \mt{Z} \|_1 = \sum_{k=1}^d \left| \lambda_k (\mt{Z}) \right|$ corresponds to the $\ell_1$-norm of the vector of eigenvalues of $\mt{Z}$. 
Analytic reconstruction guarantees for low rank matrix reconstruction via \eqref{eq:phaselift} have been established for $m = C r d \mathrm{polylog}(d)$ sufficiently random measurements \cite{candes_power_2010,recht_guaranteed_2010,gross_recovering_2011,liu_universal_2011, foucart_mathematical_2013}.

Phase retrieval---i.e.\ the problem of recovering a complex vector $\vc{x} \in \mathbb{C}^d$ from measurements of the form \eqref{eq:phaseless_measurements}---can also be re-cast as a particular instance of matrix recovery:
\begin{equation}\label{eq:lifting}
y_k = \left| \langle \vc{a}_k, \vc{x} \rangle \right|^2 + \epsilon_k
= \tr \left( \vc{a}_k \vc{a}_k^* \vc{x} \vc{x}^* \right) + \epsilon_k.
\end{equation}
This matrix formulation---expressing the quadratic relations on $x$ in terms of linear relations on its outer product $\mt{X}=\vc{x}\vc{x}^*$---is called a \emph{lifting} \cite{balan_painless_2009, candes_phase_2013}. 
Because the unknown quantity is now represented by a rank-one matrix,
$\mt{X} = \vc{x} \vc{x}^*$, it is natural to  use the convex reconstruction protocol \eqref{eq:phaselift} for recovery.
This approach is called \emph{PhaseLift} and it has been shown that $m = C n \log (n) $ random Gaussian measurement vectors $\vc{a}_1,\ldots,\vc{a}_m \in \mathbb{C}^d$ suffice to guarantee that PhaseLift
recovers an unknown vector $\vc{x} \in \mathbb{C}^d$, up to a global phase factor, with high probability \cite{candes_exact_2013}.

The matrices $\mt{X}= \vc{x} \vc{x}^*$ associated with PhaseLift are not only rank-one, but also positive semidefinite: $\mt{X} \geq \mt{0}$.
Using this additional constraint, one can reduce PhaseLift to a feasibility problem \cite{demanet_stable_2014,candes_solving_2013}.
For instance, already $m = C n$ Gaussian measurements suffice to reconstruct any $\vc{x} \in \mathbb{C}^d$ via solving
\begin{equation} \label{eq:nnls}
	  \mt{Z}^\sharp = \underset{\mt{Z} \geq \mt{0}}{\textrm{argmin}} \left\| \mathcal{A}(\mt{Z}) - \vc{y} \right\|_{\ell_q}
\end{equation}
with $q=1$ \cite{candes_solving_2013}. 
This reconstruction has an added benefit: it does not require an a priori noise bound $\eta$ as additional input. The reconstruction error---measured in Frobenius norm $\| \mt{Z} \|_2 =\sqrt{ \tr \left( \mt{Z}^2 \right)}$---scales directly proportional to the true noise level \cite{candes_solving_2013}:
\begin{equation*}
\left\| \mt{Z}^\sharp - \vc{x}\vc{x}^* \right\|_2 \leq C_2 \frac{ \| \error \|_{\ell_1}}{ m}.
\end{equation*}
Going to back from matrices to vectors,
there exists a global phase $\phi \in [0,2 \pi[$ such that the largest eigenvector $\vc{z}^\sharp$ of $\mt{Z}^\sharp$ obeys
\begin{equation*}
\left\| \vc{z}^\sharp - \mathrm{e}^{i \phi} \vc{x} \right\|_{\ell_2} \leq C_2 \min \left\{ \| \vc{x} \|_{\ell_2}, \frac{ \| \error \|_{\ell_1}}{m \| \vc{x} \|_{\ell_2}} \right\}
\end{equation*}
c.f.\ 
 \cite[Theorem 1.3]{candes_solving_2013}.

These PhaseLift results can be generalized to cover recovery of hermitian matrices with higher rank. For instance, \cite[Theorem 2]{kueng_low_2015} implies that with high probability $m \geq C r n$ random Gaussian measurements $\mt{A}_k = \vc{a}_k \vc{a}_k^*$ suffice to reconstruct any hermitian rank-$r$ matrix $\mt{X}$ via \eqref{eq:phaselift}. 
If $\mt{X} \in H_d$ is also positive semidefinite, then reconstruction may be done via \eqref{eq:nnls} with $1 \leq q \leq \infty$ \cite[Theorem 4]{kabanava_stable_2015}. 
This reconstruction is not only stable with respect to noise corruption, but also robust towards the model assumption of low rank. The minimizer $\mt{Z}^\sharp$ of \eqref{eq:nnls} obeys
\begin{equation*}
\| \mt{Z}^\sharp - \mt{X} \|_2 \leq \frac{ C_2}{\sqrt{r}} \sigma_r \left( \mt{X} \right)+C_3 \frac{ \| \error \|_{\ell_q}}{m^{\frac{1}{q}}},
\end{equation*}
where 
\begin{equation}
\sigma_r \left( \mt{X} \right) = \inf \left\{ \| \mt{X} - \mt{Z} \|_1: \; \mathrm{rank}(\mt{Z}) = r \right\} \label{eq:approximation_error}
\end{equation}
 is the nuclear norm error of the best rank-$r$ approximation to $\mt{X}$.

\section{Results: Low-rank recovery from Clifford orbits}

In this work, we prove guarantees for low-rank recovery from Clifford orbits.

More precisely, set $d=2^n$ for some $n$, fix $\vc{z}\in\CC^d$ with $\|z\|_{\ell_2}=1$ and let $\vc{z}\vc{z}^*$ be the projection operator onto $z$.
The Clifford orbits we are concerned with are the sets
\begin{equation*}
	\Cli_n \cdot \vc{z}\vc{z}^*
	=
	\{
		\mt{U} \vc{z}\vc{z}^* \mt{U}^\dagger
		\,|\, \mt{U} \in \Cli_n
	\}.	
\end{equation*}
Sometimes, we find it advantageous 
to talk about the vectors $\mt{U}\vc{z}$, rather than their projections.
The two points of view are consistent if we set
\begin{equation*}
	\Cli_n \cdot \vc{z}
	= \{ \mt{U} \vc{z} \,|\, \mt{U} \in \Cli_n \} / \{e^{i\phi}\}_\phi,
\end{equation*}
where the quotient means that if two vectors differ by a phase $\mt{U}_1 \vc{z}= e^{i\phi} \mt{U}_2\vc{z}$, we retain only one of them (equivalently: we work in projective space). 
These orbits are always finite.

The quality of the recovery guarantee depends on a measure of sparsity of the expansion coefficients of $\vc{z}\vc{z}^*$ with respect to the Pauli basis.
To state that measure, consider a general hermitian matrix $\mt{Z}$.
Then
\begin{equation*}
	\mt{Z} = \sum_{a\in\ZZ_2^n} 2^{-n/2} (\tr \mt{W}_a \mt{Z})\,\mt{W}_a.
\end{equation*}
Thus, up to normalization constants, the expansion coefficients are given by
the \emph{characteristic function}\footnote{
	To explain the origina of the terminology,
	recall that in probability theory, the \emph{characteristic function} is the Fourier transform of a probability distribution. The quantum analogue of a distribution is a \emph{density operator} and the expansion of a density operator with respect to the Weyl-Heisenberg group plays an analogues role, in quantum probability theory, to the classical characteristic function (cf.\ e.g.\ \cite{walls_quantum_2008,gross_hudson_2006})
} 
$\Xi(Z): \ZZ_2^{2n}\to\mathbb{R}$ of $Z$. 
It is defined as
\begin{equation}\label{eq:characteristic}
	\Xi (Z)(a):= \tr \left( \mt{W}_a \mt{Z} \right), \quad a \in \ZZ_2^{2n}.
\end{equation}
(The same function is called the \emph{spreading function} in time-frequency analysis \cite{shenoy_weyl_1994, kozek_spectral_1996, grochenig2013foundations, pfander-chapter}.)

Our bounds turn out to depend on the 
$\ell_4$-norm of this function:
\begin{equation}
	\left\|\Xi \left( \vc{z} \vc{z}^*\right)\right\|_{\ell_4}
	=
	\Big(
		\sum_{a \in \ZZ_2^n} \big(\tr \mt{W}_a \vc{z} \vc{z}^*\big)^4
	\Big)^{1/4}. \label{eq:characteristic_function}
\end{equation}
The definition of the Clifford group as the normalizer of the Paulis implies that this quantity is constant along Clifford orbits.
Smaller values of $\|\Xi \left( \vc{z}\vc{z}^*\right)\|_{\ell_4}$ turn out to lead to better recovery guarantees.
At the same time,
the number of non-zero coefficients of the characteristic function is lower-bounded by 
$d^2\,\|\Xi \left( \vc{z}\vc{z}^*\right)\|_{\ell_4}^{-4}$.
In this sense, Clifford orbits are connected with good recovery guarantees only if their elements have a ``spread out'' or ``dense'' characteristic function.

More precisely, the following quantity
\begin{equation}
\kappa (\vc{z},r ) := \left( \frac{r}{d} \left\| \Xi \left( \vc{z} \vc{z}^*\right)\right\|_{\ell_4}^4 + 1 \right)^2. \label{eq:kappa}
\end{equation}
appears in the sampling rate of our main result.

\begin{theorem}[Main Theorem, general version]\label{thm:main}
  Let $d=2^n$ for some $n\in\mathbb{N}$, $\vc{z}\in\CC^d$ with $\|z\|_{\ell_2}=1$, and $1\leq r \leq d$.
  Choose
  \begin{equation}
  	m \geq C_1 \kappa (\vc{z},r) r d \log (d) \label{eq:sampling_rate}
  \end{equation}
  measurements $\mt{A}_k = \vc{a}_k \vc{a}_k^*$ independently and uniformly at random from the Clifford orbit $\Cli_n \cdot \vc{z}\vc{z}^*$.
  Then with probability at least $1-\mathrm{e}^{-\frac{\gamma m}{\kappa (\vc{z},r)}}$ any hermitian rank-$r$ matrix can be recovered from these measurements in the following sense:

  For every hermitian rank-$r$ matrix $\mt X$ and every $q\in[1,\infty]$, the minimizer $Z^\sharp$ of the convex optimization problem \eqref{eq:phaselift}  fulfills
\begin{equation}
  \| \mt{Z}^\sharp - \mt{X} \|_2 
  \leq    \frac{C_2}{\sqrt{r}} \sigma_r (\mt{X}) + C_3 \sqrt{ \kappa (\vc{z},r)} d m^{-\frac{1}{q}} \eta
  \label{eq:stability}
\end{equation}
  where $\eta$ is the noise bound from \eqref{eq:phaselift} and the approximation error $\sigma_r (\mt{X})$ was defined in \eqref{eq:approximation_error}. 
  Here, $C_1,C_2,C_3$ denote sufficiently large constants and the constant $\gamma$ is sufficiently small.
\end{theorem}

Note that we have normalized the measurement vectors such that $\|a_k\|_{\ell_2}=1$. 
Other normalization conventions are also common.
For example, Gaussian random measurements have an expected length of $\|a_k\|_{\ell_2}\simeq\sqrt{d}$.
If we we drop this normalization restriction, the noise bound in (\ref{eq:stability}) generalizes to 
\begin{equation*}
  C_3 \sqrt{ \kappa \left(\vc{z}/\| \vc{z} \|_{\ell_2},r\right)} \frac{d}{\| \vc{z} \|_{\ell_2}^2} m^{-\frac{1}{q}} \eta.
\end{equation*}
Using the Gaussian normalization, the noise bound becomes independent of the ambient dimension $d$.

We prove Theorem~\ref{thm:main} by following techniques presented in \cite{kabanava_stable_2015}: we establish a strong notion of a matrix-valued null space property---see Definition~\ref{def:nsp} below---by invoking Mendelson's Small Ball Method \cite{koltchinskii_bounding_2015, mendelson_learning_2015,tropp_convex_2015}.
In order to do so, we employ recent insights about the fourth moments of the Clifford group. These are described in our companion paper \cite{other_paper}.

Theorem~\ref{thm:main} depends on the parameter $\kappa (\vc{z},r)$. According to Ref.~\cite{other_paper} it obeys 
\begin{equation}
\left(\frac{r}{d+1}+1\right)^2 \leq \kappa (\vc{z}, r ) \leq (r+1)^2
\label{eq:kappa_bound}
\end{equation}
and the upper bound is saturated for stabilizer states. 
Thus, Theorem~\ref{thm:main} requires a sampling rate of
\begin{equation*}
m \geq 2 C_1 r^3 d \log (d)
\end{equation*}
for randomly chosen stabilizer state measurements.
We believe that this worst case scaling in the rank parameter $r$ is an artifact of the proof technique. 
In contrast to this, typical orbits $\mathrm{Cl}(\vc{z})$ obey \cite{other_paper}
\begin{equation*}
\kappa(\vc{z},r) \leq \frac{6r}{d+1}+1 \leq 7
\end{equation*}
and the impact of $\kappa (\vc{z},r)$ on the sampling rate \eqref{eq:sampling_rate} is negligible.

As explained in Section~\ref{sec:convex}, if the matrices $\mt{X}$ are assumed to be positive semi-definite, one can sometimes use the convex otpimization problem (\ref{eq:nnls}) instead of (\ref{eq:phaselift}) for recovery. 
The most obvious advantage is that (\ref{eq:nnls}) does not require an estimate for the strength of the noise vector $\error$. 
Our second main result makes this precise for Clifford orbit measurements.

\begin{theorem}[Main Theorem, PSD version]\label{thm:main_psd}
	The statements of Theorem~\ref{thm:main} continue to hold under the following substitutions:
	\begin{itemize}
		\item
		The lower bound on the probability of success is weakened to $1-(d+1) \mathrm{e}^{-\frac{ \gamma m}{d+1}}$. 
  		\item
		The statement ranges over Hermitian matrices $\mt X$ of rank $r$, which are in addition assumed to be positive-semidefinite.
		\item
		The reconstruction $Z^\sharp$ is now given by the minimizer of the convex optimization problem in Eq.~(\ref{eq:nnls}),  for an arbitrary choice of $q\in[1,\infty]$.
  		\item
		The number $\eta$ in (\ref{eq:stability}) is replaced by $\|\error\|_{\ell_q}$, the true noise strength,
		with $q$ the same as above.
  \end{itemize}
\end{theorem}

Explicit bounds on the constants  $C_1, C_2, C_3, \gamma$ that appear in the two main theorems can in principle be extracted from our proofs. 
If one is only interested in the statement of Theorem~\ref{thm:main} alone (as opposed to both Theorem~\ref{thm:main} and Thoerem~\ref{thm:main_psd}), the constants improve somewhat.

Phase retrieval via PhaseLift is a particular case of matrix reconstruction, where $\mt{X} = \vc{x} \vc{x}^*$ is both positive semidefinite and rank-one. 
In this case, the bound in \eqref{eq:kappa_bound} becomes $\kappa (\vc{z},1) \leq 4$ and Theorem~\ref{thm:main_psd} implies the following statement:

\begin{corollary}[PhaseLift with Clifford orbit measurements]
  \label{cor:phaselift}
  Let $d=2^n$ and $\vc{z}\in\CC^d$ with $\|\vc{z}\|_{\ell_2} = \sqrt d$. 
  Choose
  \begin{equation}\label{eq:phaseLiftm}
  	m \geq  4 C_1\, d \log (d) 
  \end{equation}
  vectors $\vc{a}_1, \dots, \vc{a_m}$ uniformly and independently at random from Clifford orbit $\Cli_n\cdot \vc{z}$.

  Then with probability at least $1-d \mathrm{e}^{-\frac{\gamma m}{4d}}$, the phase retrieval problem \eqref{eq:phaseless_measurements} is well-posed in the following sense:
  For every $\vc{x}\in\CC^d$ and every $q\in[1,\infty]$, the leading eigenvector $\vc{z}^\sharp$ of the minimizer $\mt{Z}^\sharp$ of Eq.~(\ref{eq:nnls}) fulfills
  \begin{equation*}
		\min_{\phi\in[0,2\pi)} \|\vc{z}^\sharp - \mathrm{e}^{i \phi} \vc{x}\|_{\ell_2} \leq 2 C_3 \min \left\{ \| \vc{x} \|_{\ell_2}, \frac{ \| \error \|_{\ell_q}}{m^{1/q} \| \vc{x} \|_{\ell_2}} \right\}.
  \end{equation*}
\end{corollary}

Up to a single $\log$-factor in the sampling rate $m$, and a weaker bound on the probability of failure, Corollary~\ref{cor:phaselift} reproduces \cite[Theorem 1.3]{candes_solving_2013}--the strongest recovery guarantee for PhaseLift with Gaussian measurements we are aware of.
We have chosen the particular normalization $\| \vc{z} \|_{\ell_2}= \sqrt{d}$ to match the typical scaling of Gaussian random vectors and facilitate a direct comparison with Ref.~\cite{candes_solving_2013}.
Corollary~\ref{cor:phaselift} provides a theoretical justification for our prior numerical observation that random stabilizer measurements show close-to-optimal behavior as measurements for phase retrieval \cite[Section 2]{gross_partial_2014}.

We do not know whether the $\log$-factor in Eq.~(\ref{eq:phaseLiftm}) is necessary or not.
In the case of Pauli measurements, one can prove the necessity of a $\log$-factor by considering the recovery of stabilizer states \cite[Theorem 15]{gross_recovering_2011}. 
This suggests trying to establish a lower bound on $m$ by analyzing the recovery of stabilizer states under stabilizer measurements. 
However, heuristic arguments (Appendix~\ref{sec:heuristic}) indicate that this strategy is bound to fail. 
We therefore leave it as an open problem to determine whether or not the scaling of the sampling rate $m$ can be made linear in the dimension $d$, or whether a logarithmic correction is required.

\section{Moments of group orbits}
\label{sec:moments}

Our results make use of representation theory to proof recovery guarantees for measurements that are sampled from group orbits.
While we apply it only to the Clifford group, the technique is more general than that.
In this section, we introduce the underlying concepts.

In our analysis of the probabilistic construction of the measurement operator $\mathcal{A}$, we will make essential use of finite moments of the random vectors $a_k$.
Here, we define the $2t$-th moments of a random vector $a \in \mathbb{C}^d$ as
\begin{equation}\label{eq:moments}
  \mathbb{E} \left[ \vc{a}^{\otimes t}(\vc{a}^*)^{\otimes t} \right]
\end{equation}
(we will not make use of moments of odd degree.)
To give an example, we consider the cases where $a$ is a random Gaussian, or a vector drawn uniformly from the unit-sphere in $\CC^d$. 
In these cases, there is a simple, explicit expression for the moments.
Indeed, let $S_t$ be the symmetric group on $t$ symbols and consider its representation on $(\CC^d)^{\otimes t}$ by permuting tensor factors:
\begin{equation*}
	\pi \, x_1 \otimes \dots \otimes x_t = x_{\pi_1} \otimes \dots \otimes x_{\pi_t},
	\quad
	\forall x_k \in \CC^d, \pi \in S_t.
\end{equation*}
Let $\Sym^t(\CC^d)$ be the \emph{totally symmetric subspace}, i.e.\ the subspace of $(\CC^d)^{\otimes t}$ on which $S_t$ acts trivially. 
Let $P_{[t]}$ be the orthogonal projection onto it.
It is clear that $a^{\otimes t}\in \Sym^t(\CC^d)$ with probability one.
Therefore, the $2t$-th moments as defined in Eq.~(\ref{eq:moments}) is a Hermitian operator with support in $\Sym^t(\CC^d)$.
For Gaussian or uniform random vectors, one can show (see below) that it is, in fact, proportional to the projection onto the totally symmetric subspace:
\begin{equation}
  \mathbb{E} \left[ \vc{a}^{\otimes t}(\vc{a}^*)^{\otimes t} \right]
  = c_t\,P_{[t]}, \label{eq:infty_design}
\end{equation}
with normalization constant
\begin{equation*}
	c_t
	=
	\frac{\mathbb{E}[\|a\|_{\ell_2}^{2t}]}{\dim  \Sym^t(\CC^d)}.
\end{equation*}

Often, one can cast the analysis of randomized constructions into a form that only makes use of $2t$-th moments up to some finite value of $t$.
This has been used in particular for the analysis of PhaseLift \cite{gross_partial_2014,kueng_low_2015,ehler_phase_2015}.
The strongest result, established in Ref.~\cite{kueng_low_2015}, shows that any random vector whose $4$-th moments match the ones of vectors drawn uniform from the sphere, performs essentially optimally for PhaseLift.
Such ensembles have a name \cite{delsarte_spherical_1977, renes_symmetric_2004}:

\begin{definition}
	A random vector $a \in \mathbb{C}^d$ taking values on the complex unit-sphere is called a \emph{complex projective $t$-design} if its $2t$-th moment is proportional to the projection $P_{[t]}$ onto the totally symmetric subspace.
\end{definition}

With the relevance of moment calculations for phase retrieval established, it is natural to ask how to identify natural random vectors whose moments can be computed.
A simple but powerful approach to this problem is to relate moments to symmetries \cite{bannai_survey_2009}.

Indeed, assume that $a$ is drawn from some set $S\subset \CC^d$. 
Let $G$ be  a subgroup of the unitary group $U(d)$ such that the distribution on $S$ is $G$-invariant (if $S$ is finite and $a$ is drawn uniformly from $S$, this just means that $G$ acts on $S$).
Then clearly, for every $U\in S$, the operator $U^{\otimes t}$ commutes with 
$\mathbb{E} \left[ \vc{a}^{\otimes t}(\vc{a}^*)^{\otimes t} \right]$.
This allows us to invoke Schur's Lemma: 
Let
\begin{equation*}
	\Sym^t(\CC^d) = \bigoplus_{\lambda} V_\lambda \otimes \CC^{d_\lambda}
\end{equation*}
be the decomposition of $\Sym^t(\CC^d)$ into irreps $V_\lambda$ of $G$ with multiplicity $d_\lambda$. 
Then Schur's Lemma says that
\begin{equation}\label{eq:moments_schur}
  \mathbb{E} \left[ \vc{a}^{\otimes t}(\vc{a}^*)^{\otimes t} \right]
  =
  \bigoplus_{\lambda}
  	P_\lambda \otimes B_\lambda,
\end{equation}
where $P_\lambda$ is the identity on $V_\lambda$ and $B_\lambda$ a suitable matrix acting on the multiplicity space $\CC^{d_\lambda}$.
If all irreps are non-degenerate, i.e.\ $d_\lambda = 1\,\forall \lambda$, the expression simplifies to 
\begin{equation}\label{eq:moments_nondeg}
  \mathbb{E} \left[ \vc{a}^{\otimes t}(\vc{a}^*)^{\otimes t} \right]
  =
  \bigoplus_{\lambda} \beta_\lambda P_\lambda,
\end{equation}
for suitable $\beta_\lambda\in\CC$.
This turns out to be the case for the groups we are interested in.

This analysis allows us to give a one-line proof of Eq.~(\ref{eq:infty_design}):
The formula follows from the fact that the Gaussian and the Haar distribution are invariant under $U(d)$ and that $U(d)$ acts irreducibly on $\Sym^t(\CC^d)$ for every $t,d$.

More generally: 
Let $G\subset U(d)$ be a group such that $G$ acts irreducibly on $\Sym^4(\CC^d)$. 
Then by the above discussion, the orbit of any normalized vector $z\in\CC^d$ under $G$ is a complex projetive 4-design and therefore, Ref.~\cite{kueng_low_2015} establishes near-perfect recovery guarantees for PhaseLift with measurements sampled from this orbit.
(Ref.~\cite{gross_partial_2014} gives weaker results, but makes non-trivial statements already for $3$-designs.)

The analysis in Ref.~\cite{kueng_low_2015} uses 4-th moments to establish certain large-deviation bounds on quantities associated with the random vectors.
If a random vector $a$ fails to be a 4-design, these arguments cannot be used directly.
However, the technical premise of this paper is that the proofs can sometimes be adapted.
Indeed, if the measurement ensemble affords a symmetry group that
is sufficiently large such that $\Sym^4(\CC^d)$ decomposes into few, simple, and ideally non-degenerate representation spaces, then it might be feasible to bound all necessary quantities from Eqs.~(\ref{eq:moments_schur}), (\ref{eq:moments_nondeg}).

Motivated by this, the present authors were part of a collaboration that computed the representation theory of the Clifford group acting on $\Sym^4(\CC^d)$ \cite{other_paper}.
While it was already known that the action could not be irreducible \cite{zhu_multiqubit_2015, webb_clifford_2015, kueng_qubit_2015}, the ``second best'' scenario turned out to be realized: 
There are only two, non-degenerate irreducible representations. 
What is more, there is a simple description of these irreps.
The results from \cite{other_paper} required in this paper are summarized in the following theorem:

\begin{theorem}[\cite{other_paper}] \label{thm:main_technical}
Let $d=2^n$ and let $\mt{P}_{[4]}$ denote the projector onto the totally symmetric subspace $\Sym^4(\CC^d)$.
Let
\begin{equation*}
	\mt{Q} = \frac{1}{d^2}\sum_{a \in \ZZ_2^n} \mt{W}_a^{\otimes 4}
\end{equation*}
and define
\begin{equation}
	P_+=P_{[4]}Q,\qquad P_-=P_{[4]}(1-Q).
\end{equation}
Then $\{P_+, P_-\}$ are the projections onto the irreducible representations of $\Cli_n$ within $\Sym^4(\CC^d)$.

In particular, if $a$ is drawn uniformly from a Clifford orbit $\Cli_n\cdot z$ with $\|z\|_{\ell_2}=1$, it holds that
\begin{equation}
  \mathbb{E} \left[ \vc{a}^{\otimes t}(\vc{a}^*)^{\otimes t} \right]
	= \beta_+ (z) P_+ + \beta_- (z) P_-,\label{eq:main_technical}
\end{equation}
with coefficients
\begin{align*}
\beta_+ (\vc{z} ) =& \frac{6}{(d+2)(d+1)d^2} \left\| \Xi \left( \vc{z} \vc{z}^*\right)\right\|_{\ell_4}^4, \\
\beta_- (\vc{z} ) =& \frac{24 \left( 1 - \frac{1}{d^2}\left\| \Xi \left( \vc{z} \vc{z}^*\right)\right\|_{\ell_4}^4 \right)}{(d+4)(d+2)(d+1)(d-1)}.
\end{align*}
\end{theorem}

The dependency of the coefficients $\beta_\pm(z)$ on the characteristic function is the ultimate reason for $\Xi \left( z z^*\right)$ appearing in the sampling rate $m$ of Theorem~\ref{thm:main} through $\kappa(z,r)$.
To compare the situation to 4-designs, we borrow the bound
\begin{equation}
\frac{2d}{(d+1)} \leq \left\| \Xi \left( \vc{z} \vc{z}^*\right)\right\|_{\ell_4}^4 \leq d \;\; \forall z \in \mathbb{C}^d: \|z \|_{\ell_2}=1 \label{eq:alpha_bound}
\end{equation}
from Ref.~\cite{other_paper}.
It includes the special case $\| \Xi \left( \vc{z} \vc{z}^*\right)\|_{\ell_4}^4= \frac{4d}{(d+3)}$, which is indeed attained for certain $z$'s \cite{other_paper}.
One verifies that for this value, the coefficients coincide: $\beta_+ (z) = \beta_- (z) = \binom{d+3}{4}^{-1}$.
Using $P_+ + P_- = P_{[4]}$, this implies that Eq.~\eqref{eq:main_technical} reduces to the the defining property of a 4-design. 
Hence, the strong recovery results from \cite{kueng_low_2015,kabanava_stable_2015} apply to these specific orbits.

However, for general orbits, $\beta_+(z)$ and $\beta_- (z)$ do not coincide.
Stabilizer states are an extreme case, in the sense that they saturate the upper bound presented in \eqref{eq:alpha_bound} \cite{other_paper}, which in turn implies that the difference between $\beta_+ (z)$ and $\beta_-(z)$ is maximal.
For such orbits, we obtain the weakest results, in that the oversampling factor $\kappa(z,r)$ in Theorem~\ref{thm:main} becomes largest.
Fortunately, the averse scaling of $\kappa(z,r)$ does not become relevant for bounded rank $r$. 
In particular, our results on phase retrieval ($r=1$) are near-optimal for all Clifford orbits (c.f.~Corollary~\ref{cor:phaselift}).

Finally, we mention again that the deviation between the moments of Clifford orbits and of uniform random vectors first occurs for $t=4$. 
Clifford orbits have recently been proven \cite{zhu_multiqubit_2015,webb_clifford_2015,kueng_qubit_2015} to form are known to form complex projective $t$-designs for $t=1,2,3$.

\section{Proofs}

 Our proof strategy is as follows: We aim to establish a robust \emph{Null Space Property} for measurement operators $\mathcal{A}$ comprised of random elements of a Clifford orbit. Roughly speaking, this means that no low-rank matrix is contained in the kernel (or null space) of $\mathcal{A}$. To do so, we follow the proof technique of Ref.~\cite{kabanava_stable_2015} and invoke a now-well-known tool called \emph{Mendelson's Small Ball Method} \cite{koltchinskii_bounding_2015,mendelson_learning_2015,tropp_convex_2015}. 
This statement depends on certain concentration properties of the measurements $A_k = a_k a_k^*$. These, in turn, are derived from representation theoretic data of the Clifford group, most notably Theorem~\ref{thm:main_technical}.

The remainder of this section is organized as follows:
\begin{enumerate}
\item In Section~\ref{sub:nsp} we recall the definition of the Null Space Property, as well as Mendelson's Small Ball Method.
\item Section~\ref{sub:auxiliary} shows how relevant concentration parameters of the measurements can be derived from the representation theory of their symmetry group.
\item The main work is done in Section~\ref{sub:clifford_nsp}, where we combine these ingredients to prove a Null Space Property for Clifford orbits.
\item In Section~\ref{sub:main_proof} we use this Null Space Property to derive our first main result: Theorem~\ref{thm:main}.
\item Finally, Section \ref{sub:psd} generalizes our findings to positive-semidefinite matrix recovery (Theorem~\ref{thm:main_psd}).
\end{enumerate}

\subsection{The robust Null Space Property and Mendelson's Small Ball Method} \label{sub:nsp}

The notion of a \emph{Null Space Property} is somewhat folklore in the field of compressed sensing, see e.g.\ \cite{foucart_mathematical_2013} for a discussion of its origin. One can define analogous properties for matrix reconstruction \cite{mohan_iterative_2010, recht_null_11, recht_necessary_2008, fornasier_low_2011, kabanava_stable_2015}.
Roughly speaking, a measurement operator $\mathcal{A}: H_d \to \mathbb{R}^m$ obeys a null space property of order $r$, if no rank-$r$ matrix is contained in the kernel, or nullspace, of $\mathcal{A}$.
This is a necessary criterion for \emph{uniform} rank-$r$ matrix recovery, where uniform means that all matrices of rank $r$ or less, can be reconstructed:

\begin{definition}[Definition 3.1 in \cite{kabanava_stable_2015} for hermitian matrices] \label{def:nsp}
For fixed $r$ and $q \geq 1$, a measurement operator $\mathcal{A}:H_d \to \mathbb{R}^m$ obeyes the \emph{$\ell_q$-robust Null Space Property of order $r$} ($r/\ell_q$-NSP)
with constants $\rho \in (0,1)$ and $\tau > 0$, if
\begin{equation}
\| \mt{Z}_r \|_2 \leq \frac{\rho}{\sqrt{r}} \| Z_c \|_1 + \tau \| \mathcal{A}(\mt{Z}) \|_{\ell_q} \quad \forall \mt{Z} \in H_d.
\label{eq:nsp}
\end{equation}
Here $\mt{Z}_r=\mathrm{argmin}_{\mathrm{rank}(Y)=r} \left\| Y-Z \right\|_1$ denotes the minimizer of the approximation error $\sigma_r (\mt{Z})$ in \eqref{eq:approximation_error} and $Z_c = Z-Z_r$ obeys $\| Z_c \|_1 = \sigma_r (Z)$.
\end{definition}

Validity of a $r/\ell_q$-NSP implies that any matrix $\mt{Z}$ with rank at most $r$ obeys $\| \mt{Z} \|_2 \leq \tau \| \mathcal{A}(\mt{Z}) \|_{\ell_q}$ and therefore does not lie in  $\mathcal{A}$'s null space. 
While this is clearly necessary for uniform rank-$r$ matrix recovery, it is also sufficient, c.f.\ \cite[Theorem 3.3]{kabanava_stable_2015}. 
We will use this assertion to derive Theorem~\ref{thm:main} in Section~\ref{sub:main_proof}.

Note that \eqref{eq:nsp} is invariant under scaling and we may set $\|\mt{Z} \|_2 = 1$ without loss of generality.
Moreover, any normalized matrix $\mt{Z} \in H_d$ which also obeys $\| \mt{Z}_r \|_2 \leq \frac{\rho}{\sqrt{r}} \| \mt{Z}_c \|_1$ 
fulfills \eqref{eq:nsp}, irrespective of $\mathcal{A}$.
So, when aiming to establish a $r/\ell_q$-NSP for any particular $\mathcal{A}$, we may restrict our attention to 
\begin{equation}
T_{\rho,r} = \left\{ \mt{Z} \in H_d:\| \mt{Z}_r \|_2 > \frac{\rho}{\sqrt{r}} \| \mt{Z}_c \|_1, \| \mt{Z} \|_2 = 1 \right\}. \label{eq:T}
\end{equation}
In turn, $\mathcal{A}: H_d \to \mathbb{R}^m$ obeys the $r/\ell_q$-NSP with constants $\rho \in (0,1)$ and $\tau >0$, if
\begin{equation}
\inf_{\mt{Z} \in T_{\rho,r}} \left\| \mathcal{A} (\mt{Z}) \right\|_{\ell_q} \geq \frac{1}{\tau}. \label{eq:nsp_sufficient}
\end{equation}
Note that the parameters $r,\rho$ implicitly feature in the definition of $T_{\rho,r}$, while $\tau$ is inversely proportional to the best lower bound achievable in \eqref{eq:nsp_sufficient}.

Our NSP-proof is based on the following statement \cite{koltchinskii_bounding_2015,mendelson_learning_2015,tropp_convex_2015}. 

\begin{theorem}[Variant of Mendelson's small ball method\footnote{We remark that Mendelson's small ball method often refers to a lower bound on $\inf_{\vc{z} \in E}\sqrt{\sum_{k=1}^m \left| \langle  \phi_k, \vc{z} \rangle \right|^2}$, while this statement is slightly stronger, as it bounds 
$\frac{1}{\sqrt{m}} \inf_{\vc{z} \in E} \sum_{k=1}^m | \langle \phi_k, \vc{z} \rangle |$ instead. This stronger claim, however, is also implied by Mendelson's original proof, see for instance \cite[Remark 5.1]{kabanava_stable_2015}.}] \label{thm:mendelson}
Fix $E \subset \mathbb{R}^d$ and let $\phi_1,\ldots,\phi_m \in \mathbb{R}^d$ be independent copies of a random vector $\phi$.
For $\xi >0$ define
\begin{align}
Q_\xi \left( E ; \phi \right) =& \inf_{\vc{z} \in E} \mathrm{Pr} \left[ \left| \langle \phi, \vc{z} \rangle \right| \geq \xi \right], \quad \textrm{and} \label{eq:Qxi} \\
W_m (E; \phi) =& \mathbb{E} \left[ \sup_{\vc{z} \in E} \langle \vc{h}, \vc{z} \rangle \right] \quad \textrm{with} \label{eq:Wm} \\
\vc{h} =& \frac{1}{\sqrt{m}} \sum_{k=1}^m \varepsilon_k \phi_k \in \mathbb{R}^d, \label{eq:h}
\end{align}
where each $\varepsilon_k$ is an independent instance of a Rademacher random variable (i.e.\ $\varepsilon_k$ assumes $+1$ and $-1$ with equal probability). Then for any $\xi >0$ and $t \geq 0$, the following bound is true with probability at least $1-\mathrm{e}^{-2t^2}$:
\begin{equation*}
\frac{1}{\sqrt{m}} \inf_{\vc{z} \in E} \sum_{k=1}^m | \langle \phi_k, \vc{z} \rangle | \geq \xi \sqrt{m} Q_{2\xi} (E; \phi) - 2W_m (E; \phi) - \xi t.
\end{equation*}
\end{theorem}

In this work, we will 
employ the following corollary:

\begin{corollary} \label{cor:mendelson}
Fix $r$, $\rho$ and let $T_{\rho,r} \subset H_d$ be the set introduced in \eqref{eq:T}.
Suppose that $\mathcal{A}: H_d \to \mathbb{R}^m$ is a measurement operator containing $m$ independent instances of a single random matrix $\mt{A} \in H_d$ as individual measurements. 
Then for any $q \geq 1$, $\xi >0$ and $t \geq 0$ 
\begin{align}
& \inf_{\mt{Z} \in T_{\rho,r}} \| \mathcal{A} (\mt{Z}) \|_{\ell_q} \nonumber \\
\geq & m^{\frac{1}{q}-\frac{1}{2}}\left( \xi \sqrt{m} Q_{2\xi}(T_{\rho,r};\mt{A}) - 2 W_m (T_{\rho,r},\mt{A}) - \xi t \right) \label{eq:mendelson}
\end{align}
is true with probability at least $1- \mathrm{e}^{-2t^2}$. Here $Q_{2 \xi} (T_{\rho,r}; \mt{A})$ and $W_m (T_{\rho_r}, \mt{A})$ are the parameters defined in \eqref{eq:Qxi} and \eqref{eq:Wm}.
\end{corollary}

\begin{proof}
$H_d$ is a real-valued vector space isomorphic to $\mathbb{R}^{d^2}$ and we may identify each $\mt{A}_k$ with an instance $ \phi_k$ of the random vector $\mt{A} := \phi \in \mathbb{R}^{d^2} \simeq H_d$. 
Also, the Frobenius inner product $(Y,Z) = \tr (YZ)$ endows $H_d$ with an inner product.
Setting  $E=T_{\rho,r} \subset H_d \simeq \mathbb{R}^{d^2}$, where $T_{\rho,r}$ was defined in \eqref{eq:T} and applying Theorem~\ref{thm:mendelson} yields
\begin{equation*}
\inf_{\mt{Z} \in T_{\rho,r}} \frac{1}{\sqrt{m}} \sum_{k=1}^m \left| \left( \mt{A}_k, \mt{Z} \right) \right| = \inf_{\mt{Z} \in T_{\rho,r}} \frac{1}{\sqrt{m}}\| \mathcal{A} (\mt{Z}) \|_{\ell_1}.
\end{equation*}
Finally, we employ the basic norm inequality $\| \vc{z} \|_1 \leq m^{1-\frac{1}{q}} \| \vc{z} \|_{\ell_q} \; \forall \vc{z} \in \mathbb{R}^m,\; \forall q \geq 1$ (see for instance \cite[Equation A.3]{foucart_mathematical_2013})
to conclude
\begin{equation*}
\inf_{\mt{Z}\in T_{\rho,r}} \| \mathcal{A}(\mt{Z}) \|_{\ell_q} \geq m^{\frac{1}{q}-\frac{1}{2}} \inf_{\mt{Z}\in T_{\rho,r}} \frac{1}{\sqrt{m}} \| \mathcal{A} (\mt{Z}) \|_{\ell_1}.
\end{equation*}
\end{proof}

\subsection{Bounding the relevant parameters in Corollary~\ref{cor:mendelson} for Clifford orbits} \label{sub:auxiliary}

Two parameters feature prominently in Corollary~\ref{cor:mendelson}: $W_m \left( T_{\rho,r}; A \right)$ defined in \eqref{eq:Wm} and $Q_{2 \xi} \left( T_{\rho,r}; A\right)$ defined in \eqref{eq:Qxi}.
Both parameters crucially depend on the geometry of $T_{\rho,r} \subseteq H_d$ introduced in \eqref{eq:T} and the distribution of the measurement matrices $A= \vc{a} \vc{a}^*$. In our case, these are uniformly selected from a Clifford orbit $\Cli_n \cdot \vc{z} \vc{z}^*$.

A fist auxiliary statement addresses the geometry of $T_{\rho,r}$ and asserts that the \emph{effective rank} of every $T_{\rho,r}$ cannot be too large:

\begin{lemma} \label{lem:effective_low_rank}
Let $T_{\rho,r} \subset H_d$ be the set introduced in \eqref{eq:T} for some $\rho \in (0,1)$ and $1 \leq r \leq d$. Then 
\begin{equation}
\frac{ \| \mt{Z} \|_1^2}{\| \mt{Z} \|_2^2} \leq \left( \frac{\rho+1}{\rho}\right)^2 r \quad \forall \mt{Z} \in T_{\rho,r}. \label{eq:effective_low_rank}
\end{equation}
\end{lemma} 

\begin{proof}
Combining $\| \mt{Z}_r \|_1 \leq \sqrt{r} \| \mt{Z}_r \|_2$ with the defining property of $\mt{Z} \in T_{\rho,r}$ reveals
\begin{equation*}
\| \mt{Z} \|_1 \leq \| \mt{Z}_r \|_1 + \| \mt{Z}_c \|_1 \leq   \frac{\rho+1}{\rho} \sqrt{r} \| \mt{Z}_r \|_2, 
\end{equation*}
and the claim follows from $\| \mt{Z}_r \|_2 \leq \| \mt{Z} \|_2=1$.
\end{proof}

This insight allows one to bound the first parameter featuring in Corollary~\ref{cor:mendelson}:

\begin{proposition} \label{prop:Wm}
Fix $d=2^n$ and suppose that  $\mt{A} = \vc{a} \vc{a}^*$ results from choosing an element of a Clifford orbit $\Cli_n \cdot \vc{z} \vc{z}^*$ with $\|z \|_{\ell_2}=1$ uniformly at random. 
Also, fix $1 \leq r \leq d$, $\rho \in (0,1)$ and suppose $m \geq 2 d \log (d)$. Then 
\begin{equation*}
W_m (T_{\rho,r}; \vc{A} ) \leq \frac{ 6.2098 }{\rho} \sqrt{\frac{r \log (2d)}{d+1}}.
\end{equation*}
\end{proposition}

\begin{proof}
This proof closely resembles a comparable analysis provided in \cite{kueng_low_2015}. Matrix Hoelder together with Lemma~\ref{lem:effective_low_rank} implies
\begin{align*}
W_m (T_{\rho,r}; \mt{A}) & =\mathbb{E} \left[ \sup_{\mt{Z} \in T_{\rho,r}} \left( \mt{H}, \mt{Z} \right) \right]
\leq \sup_{\mt{Z} \in T_{\rho,r}} \| \mt{Z} \|_1 \mathbb{E} \left[ \| \mt{H} \|_\infty \right] \\
& \leq \frac{\rho+1}{\rho} \sqrt{r} \mathbb{E} \left[ \| \mt{H} \|_\infty \right]
\leq \frac{2}{\rho} \sqrt{r} \mathbb{E} \left[ \| \mt{H} \|_\infty \right],
\end{align*}
with $\mt{H} = \frac{1}{\sqrt{m}} \sum_{k=1}^m \epsilon_k \vc{a}_k \vc{a}_k^*$. Each $\mt{A}_k =\vc{a}_k \vc{a}_k^*$ obeys $\mathbb{E} \left[ A_k \right] = \frac{1}{d} \mathbb{I}$, because it is uniformly chosen from a Clifford orbit (Formula~\eqref{eq:infty_design} for $t=1$).
This property alone together with the Rademacher randomness in $H$ allows for bounding $\mathbb{E} \left[ \left\| H \right\|_{\infty} \right]$ by combining a non-commutative Khintchine inequality with a matrix Chernoff bound, see for instance \cite[Proposition 13]{kueng_low_2015}.
Adapting said statement to unit normalization  ($\| \vc{a}_k \|_{\ell_2} = \| \vc{z} \|_{\ell_2}= 1$) implies
\begin{equation*}
\mathbb{E} \left[ \| \mt{H} \|_\infty \right] \leq 3.1049  \sqrt{\frac{\log(2d)}{d+1}},
\end{equation*}
provided that $m \geq 2d \log (d)$
and the claim readily follows.
\end{proof}

Establishing a lower bound on the remaining parameter $Q_{\xi} \left( T_{\rho, r}; \vc{a} \vc{a}^* \right)$ for Clifford orbits is considerably more challenging.
We do so by applying a
Paley-Zygmund inequality
that depends on the following auxiliary statement.

\begin{lemma} \label{lem:moments}
Fix $\mt{Z} \in H_d$ and define the random variable $S_\mt{Z}:= \langle \vc{a}, \mt{Z} \vc{a} \rangle$, where $\vc{a}$ is uniformly chosen from a Clifford orbit $\Cli_n \cdot z z^*$ with $\| z \|_{\ell_2}=1$. Then
\begin{align}
\mathbb{E} \left[ S_\mt{Z}^2 \right] =& \frac{ \left( \| \mt{Z} \|_2^2 + \tr (\mt{Z})^2 \right)}{(d+1)d}  \quad \textrm{and} \label{eq:second_moment}\\
\mathbb{E} \left[ S_{\mt{Z}}^4 \right] \leq&  \left( \frac{6}{d} \left\| \Xi \left( \vc{z} \vc{z}^* \right) \right\|_{\ell_4}^4 \frac{ \| \mt{Z} \|_1^2}{\| \mt{Z} \|_2^2}+13 \right) \mathbb{E} \left[ S_{\mt{Z}}^2 \right]^2.\label{eq:fourth_moment} 
\end{align}

\end{lemma}

\begin{proof}
Equation~\eqref{eq:second_moment} is a consequence of the fact that Clifford orbits obey Formula~\eqref{eq:infty_design} for $t=2$:
\begin{align*}
\mathbb{E} \left[ S_{\mt{Z}}^2 \right] =& \mathbb{E} \left[ \langle \vc{a}, \mt{Z} \vc{a} \rangle^2\right] = \mathbb{E} \left[ \tr \left( \vc{a} \vc{a}^* \mt{Z} \right)^2 \right] \\
=&  \tr \left( \mathbb{E} \left[  \vc{a}^{\otimes 2} \left( \vc{a}^* \right)^{\otimes 2} \right]  \mt{Z}^{\otimes 2} \right) \\
=& \frac{ 2 \tr \left(\mt{P}_{\sym^2} \mt{Z}^{\otimes 2} \right)}{(d+1)d}
= \frac{ \tr (\mt{Z}^2) + \tr (\mt{Z})^2}{(d+1)d}.
\end{align*}
The last equality follows from applying standard techniques from multilinear algebra, see e.g.\ \cite[Lemma 6]{gross_partial_2014}, or \cite[Lemma 17]{kueng_low_2015}.

Deriving the fourth moment bound \eqref{eq:fourth_moment} is more involved. For any $\mt{Z} \in H_d$ Theorem~\ref{thm:main_technical} implies
\begin{align*}
 \mathbb{E} \left[ S_\mt{Z}^4 \right] =&  \mathbb{E} \left[ \langle \vc{a}, \mt{Z} \vc{a} \rangle^4 \right] = \tr \left( \mathbb{E} \left[  \vc{a}^{\otimes4} \left( \vc{a}^* \right)^{\otimes 4} \right] \mt{Z}^{\otimes 4} \right) \nonumber \\
=& \beta_+ (z) \tr \left( P_+ Z^{\otimes 4} \right) + \beta_- (z) \tr\left( P_- Z^{\otimes 4} \right).
\end{align*}
We can use $P_+ = P_{[4]} Q$ and $P_- = P_{[4]}(1-Q)$ with $Q = \frac{1}{d^2} \sum_{a \in \ZZ_2^n} W_a^{\otimes 4}$ to rewrite this expression as
\begin{align}
\mathbb{E} \left[ S_\mt{Z}^4 \right] =& \left( \beta_+ (z) - \beta_- (z) \right) \tr \left( P_{[4]} Q Z^{\otimes 4} \right) \nonumber \\
+& \beta_- (z) \tr \left( P_{[4]} Z^{\otimes 4} \right). \label{eq:technical_aux1}
\end{align}

The second  trace expression can be explicitly computed, e.g.\ by adapting the argument of \cite[Lemma 17]{kueng_low_2015}: 
\begin{align*}
& 24 \left| \tr \left( \mt{P}_{\sym^4} \mt{Z}^{\otimes 4} \right) \right| \\
=& \left| \tr (\mt{Z})^4 + 8 \tr (\mt{Z}) \tr \left( \mt{Z}^3 \right) + 3 \tr \left( \mt{Z}^2 \right)^2 \right.\\
+& \left. 6 \tr \left( \mt{Z} \right)^2 \tr \left( \mt{Z}^2 \right) + 6 \tr \left( \mt{Z}^4 \right) \right| \\
\leq & 3 \left( \tr \left( \mt{Z} \right)^2 + \tr \left( \mt{Z}^2 \right) \right)^2 + 8 |\tr (\mt{Z})| \| \mt{Z} \|_2^3 + 6 \| \mt{Z} \|_2^4 \\
\leq & 3 \left( \tr (\mt{Z})^2 + \| \mt{Z} \|_2^2 \right)^2 + 4 \left( \tr (\mt{Z})^2 + \| \mt{Z} \|_2^2 \right) \| \mt{Z} \|_2^2 + 6 \| \mt{Z} \|_2^4 \\
\leq & 13 \left( \tr (\mt{Z})^2 + \| \mt{Z} \|_2^2 \right)^2.
\end{align*}
Here, we have used standard Schatten-$p$ norm inequalities such as $\tr \left( Z^3 \right) \leq \| Z \|_3^3 \leq \| Z \|_2^3$ and $\tr \left( Z^4 \right) = \| Z \|_4^4 \leq \| Z \|_2^4$ as well as the AM-GM inequality:
 $| \tr (\mt{Z})| \| \mt{Z} \|_2 \leq \frac{1}{2} \left( \tr (\mt{Z})^2 + \| \mt{Z} \|_2^2 \right)$.

We start bounding the remaining trace expression by noticing
\begin{align*}
\left| \tr \left( P_{[4]}\mt{Q} \mt{Z}^{\otimes 4} \right) \right| \leq & \tr \left( P_{[4]} \mt{Q} | \mt{Z} |^{\otimes 4} \right) \leq  \tr\left( Q |Z|^{\otimes 4} \right) \\
=&  \frac{1}{d^2} \sum_{a \in \ZZ_2^n} \left( \mt{W}_a^{\otimes 4}, | \mt{Z}|^{\otimes 4} \right) \\
=& \frac{1}{d^2} \sum_{a \in \ZZ_2^n} \left( \mt{W}_a, | \mt{Z} | \right)^4.
\end{align*}
Here $|Z| = \sqrt{ Z Z^*}$ denotes the matrix absolute value of $Z$ and the inequalities above are standard relations for positive-semidefinite matrices.
Applying matrix Hoelder and using the fact that Schatten-$p$-norms of $\mt{Z}$ and $| \mt{Z} |$ coincide by definition allows us to deduce
\begin{align*}
\frac{1}{d^2} \sum_{a \in \ZZ_2^n} \left( \mt{W}_a, | \mt{Z} | \right)^4
\leq & \frac{1}{d^2} \sum_{a \in \ZZ_2^n} \| \mt{W}_a \|_\infty^2  \| |\mt{Z}| \|_1^2 \left( \mt{W}_a, |\mt{Z}| \right)^2 \\
=& \frac{ \| \mt{Z} \|_1^2}{d^2} \sum_{a \in \ZZ_2^n} \left( \mt{W}_a, |\mt{Z}| \right)^2 
= \frac{ \| \mt{Z} \|_1^2 \| \mt{Z} \|_2^2}{d}.
\end{align*}
The second line is due to the fact that Pauli matrices are a unitary  ($\| \mt{W}_a \|_\infty = 1$)  matrix basis of $H_d$ that is orthogonal with respect to the Frobenius inner product ($(\mt{W}_a,\mt{W}_b) = d \delta_{a,b}$).

We can now move on to bound pre-factors. Formula~\eqref{eq:alpha_bound} implies 
\begin{align*}
\beta_-(z)
\leq & \frac{24 \left( 1 - \frac{2}{(d+1)d}  \right)}{(d+4)(d+2)(d+1)(d-1)} \\
=& \frac{24}{(d+4)(d+1)^2 d} 
\leq  \frac{24}{(d+1)^2 d^2},
\end{align*}
as well as
\begin{align*}
\left| \beta_+ (z) - \beta_- (z) \right|
=& \frac{ \left| d^2+3d - 4d^2/\left\| \Xi \left( \vc{z} \vc{z}^*\right)\right\|_{\ell_4}^4 \right|}{(d+4)(d-1)} \beta_+ (z) \\
\leq& \frac{d \beta_+(z)}{(d+4)} 
\leq  \frac{6 \| \Xi \left( \vc{z} \vc{z}^*\right)\|_{\ell_4}^4}{(d+1)^2 d^2}
\end{align*}
Inserting all these individual bounds into \eqref{eq:technical_aux1} implies
\begin{align*}
\mathbb{E} \left[ S_Z^4 \right]
\leq & \left| \beta_+ (z) - \beta_- (z)\right| \left| \tr \left( P_{[4]} Q Z^{\otimes 4} \right) \right| \\
+& \left| \beta_- (z) \right| \left| \tr \left( P_{[4]} Z^{\otimes 4} \right) \right| \\
\leq & \frac{6 \left\| \Xi \left( \vc{z} \vc{z}^*\right)\right\|_{\ell_4}^4 \| Z \|_1^2 \| Z \|_2^2}{(d+1)d^3}
+ \frac{13\left(\tr (Z)^2 + \| Z \|_2^2\right)^2 }{(d+1)^2d^2} \\
\leq & \left(\frac{6}{d}\left\| \Xi \left( \vc{z} \vc{z}^*\right)\right\|_{\ell_4}^4 \frac{ \| Z \|_1^2}{\| Z \|_2^2} + 13 \right)  \left( \frac{ \tr(Z)^2 + \| Z \|_2^2}{(d+1)d} \right)^2 \\
\end{align*}
and the claim follows.
\end{proof}

The pre-factor in the 4-th moment bound \eqref{eq:fourth_moment}  depends both on the characteristic function $\Xi \left( \vc{z} \vc{z}^*\right)$ and the effective rank of $\mt{Z}$. 
However, without putting further restrictions on $Z$ and the Clifford orbit, this is unavoidable up to multiplicative constants: Choosing $z = \psi \in \mathbb{C}^d$ to be a stabilizer state and setting $Z= W_a$ ($a\neq 0 \in \ZZ_2^n$) 
results in 
\begin{align*}
\mathbb{E} \left[ S_{W_a}^4 \right] =& (d+1)\mathbb{E} \left[ S_{W_a}^2 \right]^2 \\
=& \left( \frac{1}{d} \left\| \Xi \left( \psi \psi^*\right) \right\|_{\ell_4}^4 + 1 \right)\mathbb{E} \left[ S_{W_a}^2 \right]^2.
\end{align*}

\begin{proposition} \label{prop:Qxi}
Fix $d=2^n$ and let $\mt{A} = \vc{a} \vc{a}^*$ be a randomly chosen element of a Clifford orbit $\Cli_n \cdot \vc{z} \vc{z}^*$. Then the parameter $Q_{\xi} (T_{\rho,r},\mt{A})$, featuring in Corollary~\ref{cor:mendelson}, obeys
\begin{align*}
Q_{\xi} \left( T_{\rho,r}; \mt{A} \right) \geq \frac{\rho^2}{24 \sqrt{\kappa (\vc{z},r)}} \left(1-\left( \sqrt{(d+1)d} \xi \right)^2 \right)^2,
\end{align*}
where $\kappa (\vc{z},r)$ was introduced in \eqref{eq:kappa}. 
This bound is true
for any $0 \leq \xi \leq \frac{1}{\sqrt{(d+1)d}}$, $1 \leq r \leq d$ and $\rho \in (0,1)$. 
\end{proposition}

\begin{proof}
Fix $\mt{Z} \in T_{\rho,r}$, $\xi  \geq 0$ and  define the real-valued random variable $S_\mt{Z}=\mathrm{tr} \left( \vc{a} \vc{a}^* \mt{Z} \right)$, where $\vc{a}\vc{a}^*$ is chosen uniformly from $\Cli_n \cdot \vc{z} \vc{z}^*$.
Then
\begin{align*}
\mathrm{Pr} \left[ | \tr \left( \mt{A} \mt{Z} \right)| \geq \xi \right]
=& \mathrm{Pr} \left[ |S_\mt{Z}| \geq \xi \right] = \mathrm{Pr} \left[ S_\mt{Z}^2 \geq \xi^2 \right] \\
\geq & \mathrm{Pr} \left[ S_\mt{Z}^2 \geq (d+1)d \xi^2 \mathbb{E} \left[ S_{\mt{Z}}^2 \right] \right], 
\end{align*}
where the last inequality follows from \eqref{eq:second_moment}, because every $\mt{Z} \in T_{\rho,r}$ obeys $\| \mt{Z} \|_2=1$. 
Applying the Paley-Zygmund inequality, see e.g.\ \cite[Lemma 7.16]{foucart_mathematical_2013}, to the non-negative random variable $S_\mt{Z}^2$ yields
\begin{align}
& \mathrm{Pr} \left[ S_\mt{Z}^2 \geq (d+1)d \xi^2 \mathbb{E} \left[ S_\mt{Z}^2 \right] \right] \nonumber \\
\geq & \left( 1-(d+1)d \xi^2 \right)^2 \frac{ \mathbb{E} \left[ S_\mt{Z}^2 \right]^2}{\mathbb{E} \left[ S_{\mt{Z}}^4 \right]} \nonumber \\
\geq & \frac{\left( 1- (d+1)d \xi^2 \right)^2}{\frac{6}{d} \| \Xi \left( \vc{z} \vc{z}^*\right)\|_{\ell_4}^4 \frac{ \| \mt{Z} \|_1^2}{\| \mt{Z} \|_2^2}+13}, \label{eq:Qxiaux1}
\end{align}
where the last line is due to \eqref{eq:fourth_moment}. 

According to Lemma~\ref{lem:effective_low_rank}, each $\mt{Z} \in T_{\rho,r}$ admits the bound
\begin{align*}
 \frac{6}{d} \left\| \Xi \left( \vc{z} \vc{z}^*\right)\right\|_{\ell_4}^4 \frac{ \| \mt{Z} \|_1^2}{\| \mt{Z} \|_2^2}+13  
\leq & 6 \left( \frac{1+\rho}{\rho} \right)^2 \frac{r}{d} \left\| \Xi \left( \vc{z} \vc{z}^*\right)\right\|_{\ell_4}^4 +13 \\
\leq & \frac{24}{\rho^2} \left( \frac{r}{d}\left\| \Xi \left( \vc{z} \vc{z}^*\right)\right\|_{\ell_4}^4 +1 \right) \\
=& \frac{24}{\rho^2} \sqrt{\kappa (\vc{z},r)}.
\end{align*}
Inserting this into \eqref{eq:Qxiaux1} results in a lower bound that is valid for all $\mt{Z} \in T_{\rho,r}$ simultaneously. Thus it also applies to
\begin{equation*}
Q_{\xi} \left( T_{\rho,r}; \mt{A} \right) = \inf_{\mt{Z} \in T_{\rho,r}} \mathrm{Pr} \left[ \left| \left( A,Z \right) \right| \geq \xi \right],
\end{equation*}
where $(A,Z) =\tr (AZ)$
and the claim follows.
\end{proof}

\subsection{A Null Space Property for Clifford obits} \label{sub:clifford_nsp}

We have now assembled all necessary ingredients to prove a Null Space Property in the sense of Definition~\ref{def:nsp} for random Clifford orbit measurements.

\begin{theorem} \label{thm:clifford_nsp}
Set $d=2^n$, and fix $1 \leq r \leq d$, $\rho \in (0,1 )$, $1 \leq q \leq \infty$. Suppose that $\mathcal{A}:H_d \to \mathbb{R}^m$ contains
\begin{equation}
m \geq \frac{\tilde{C}_1}{\rho^6}\kappa (\vc{z},r) r d \log (2d) \label{eq:nsp_sampling_rate}
\end{equation}
randomly chosen elements of a Clifford orbit $\Cli_n \cdot \vc{z} \vc{z}^*$ with $\| \vc{z} \|_{\ell_2}=1$.
Then, with probability at least $1~-~\mathrm{e}^{-\frac{2 \rho^4 \tilde{\gamma} m}{\kappa (\vc{z},r)}}$, $\mathcal{A}$ obeys the $r/\ell_q$-NSP from Definition~\ref{def:nsp} with parameters
\begin{equation}
\rho \quad \textrm{and} \quad \tau = \frac{\tilde{C}_3}{\rho^2} \sqrt{ \kappa (z,r)} d m^{-\frac{1}{q}}. \label{eq:clifford_nsp_constants}
\end{equation}
Here, $C_1,\tilde{C}_3$ and $\tilde{\gamma}$ denote constants of sufficient size.
\end{theorem}

We point out that the sampling rate \eqref{eq:nsp_sampling_rate} scales non-optimally in the NSP parameter $\rho \in (0,1)$. The required sampling rate in comparable statements, such as \cite[Theorem 3]{kabanava_stable_2015}, only scales proportionally to $\rho^{-2}$.
This non-optimality is due to the fact that the fourth moment bound in Lemma~\ref{lem:moments} implicitly depends on  the ``effective rank'' of $\mt{Z} \in T_{\rho,r}$ which is proportional to $\frac{r}{\rho^2}$ (see Lemma~\ref{lem:effective_low_rank}).
In turn, this effective rank also features in the bound on $Q_{\xi} \left( \mt{A}; T_{\rho,r} \right)$ and affects the results of Mendelson's Small Ball Method.
We believe that such a behavior is unavoidable when using Mendelson's Small Method for Clifford orbits, and intend to address this issue in the future.

\begin{proof}
Applying Corollary~\ref{cor:mendelson} with $\xi = \frac{1}{4 \sqrt{(d+1)d}}$ and $t = \rho^2 \sqrt{\frac{\tilde{\gamma} m}{\kappa (\vc{z},r)}}$---where $\tilde{\gamma}$ is a sufficiently small constant---implies
\begin{widetext}
\begin{align}
\inf_{\mt{Z} \in T_{\rho,r}} \left\| \mathcal{A}(\mt{Z}) \right\|_{\ell_q}
\geq & m^{\frac{1}{q}-\frac{1}{2}} \left( \sqrt{m}\frac{Q_{\frac{1}{2\sqrt{(d+1)d}}}(T_{\rho,r};\mt{A})}{4\sqrt{(d+1)d}}-2W_m (T_{\rho,r};\mt{A}) - \frac{\rho^2}{4} \sqrt{\frac{ \tilde{\gamma}m}{(d+1)d \kappa (\vc{z},r)}}
\right) \nonumber \\
\geq & \frac{m^{\frac{1}{q}-\frac{1}{2}}}{\sqrt{(d+1)d}}
\left( \frac{\frac{9}{16} \rho^2 \sqrt{m}}{4 \times 24 \sqrt{\kappa (\vc{z},r )}} -  \frac{2 \times 6.2098}{\rho} \sqrt{ r d \log (2d)}  - \frac{\rho^2}{4} \sqrt{\frac{\tilde{\gamma}m}{\kappa (\vc{z},r)}} \right)  \nonumber \\
\geq &
\frac{ \rho^2 m^{\frac{1}{q}-\frac{1}{2}}}{\sqrt{(d+1)d \kappa (\vc{z},r)}} \left( \frac{  \sqrt{m}}{171}- 13 \sqrt{ \frac{ \kappa (\vc{z},r)^2}{\rho^6} r d \log (2d)}- \frac{\sqrt{\tilde{\gamma} m}}{4} \right) \label{eq:nsp_aux1} 
\end{align}
\end{widetext}
with probability at least $1-\mathrm{e}^{-\frac{2\rho^4 \tilde{\gamma} m}{\kappa (\vc{z},r)}}$. In the second line, we have inserted the bounds provided by Proposition~\ref{prop:Wm} and Proposition~\ref{prop:Qxi}, respectively.
Let us now fix
\begin{equation*}
m \geq \frac{\tilde{C}_1}{\rho^6} \kappa (\vc{z},r) rd \log (2d),
\end{equation*}
where $\tilde{C}_1$ is a sufficiently large constant. 
Note that such a choice in particular assures $m \geq 2d \log (d)$ which justifies the applicability of Proposition~\ref{prop:Wm}. Moreover,
provided that $\tilde{\gamma}$ is small enough, this choice assures that the bracket expression in \eqref{eq:nsp_aux1} is lower bounded by $\frac{2 \sqrt{m}}{\tilde{C}_3}$, where $\tilde{C}_3$ is constant.
Inserting this novel bound into \eqref{eq:nsp_aux1} allows us to conclude
\begin{equation*}
\inf_{\mt{Z} \in T_{\rho,r}} \| \mathcal{A} (\mt{Z}) \|_{\ell_q} \geq 
 \frac{ 2 \rho^2 m^{\frac{1}{q}}}{\tilde{C}_3  \sqrt{(d+1)d \kappa (\vc{z},r)}} \geq \frac{ \rho^2 m^{\frac{1}{q}}}{\tilde{C}_3d \sqrt{ \kappa (\vc{z},r)}}
\end{equation*}
with high probability. Inserting this bound into \eqref{eq:nsp_sufficient} establishes the claimed Null Space Property for $\mathcal{A}$ with probability at least $1-\mathrm{e}^{-\frac{2\rho^4 \tilde{\gamma} m}{\kappa (\vc{z},r)}}$.

\end{proof}

\subsection{Derivation of Theorem~\ref{thm:main}} \label{sub:main_proof}

Suppose that $\mathcal{A}:H_d \to \mathbb{R}^m$ obeys a $r/\ell_q$-NSP in the sense of Definition~\ref{def:nsp}. 
Then, every approximately rank-$r$ matrix $X\in H_d$ can be estimated from noisy measurements of the form $y = \mathcal{A}(X) + \epsilon$.
One way to achieve stable reconstruction is via constrained nuclear norm minimization \eqref{eq:phaselift}, provided that the parameter $\eta$ obeys $\eta \geq \| \error \|_{\ell_q}$:

\begin{theorem}[Theorem 3.3 in \cite{kabanava_stable_2015} for hermitian matrices] \label{thm:nsp_implication}
Fix $r$, $q \geq 1$ and suppose that $\mathcal{A}: H_d \to \mathbb{R}^m$ obeys a $r/\ell_q$-NSP with constants $\rho \in (0,1)$ and $\tau >0$. Then 
\begin{align}
\| \mt{Z} - \mt{X} \|_2
\leq &\frac{C_\rho}{\sqrt{r}} ( \| \mt{Z} \|_1 - \| \mt{X} \|_1 + 2 \sigma_r (\mt{X}) ) \label{eq:nsp_implication}\\
+& D_\rho \tau \| \mathcal{A}(\mt{Z}-\mt{X}) \|_{\ell_q} \quad \forall \mt{X},\mt{Z} \in H_d, \nonumber
\end{align}
with $C_\rho = \frac{(1+\rho)^2}{1-\rho}$ and $D_\rho = \frac{3+\rho}{1-\rho}$. 
\end{theorem}

Now let $X \in H_d$ be the matrix of interest and $Z^\sharp$ be the minimizer of \eqref{eq:phaselift}. 
By construction, $X$ is also a feasible point of this minimization and optimality of $Z^\sharp$ assures $\| Z^\sharp \|_1 - \| X\|_1 \leq 0$.
Moreover:
\begin{align*}
\| \mathcal{A}(\mt{X})-\mt{Z}) \|_{\ell_q} \leq & \| \mathcal{A}(\mt{X})-\vc{y} \|_{\ell_q} + \| \mathcal{A}(\mt{Z}) - \vc{y} \|_{\ell_q} \\
\leq & \| \error \|_{\ell_q} + \eta \leq 2 \eta.
\end{align*}
Inserting these inequalities into \eqref{eq:nsp_implication} implies
\begin{equation}
\| \mt{Z}^\sharp - \mt{X} \|_2 \leq \frac{2 C_\rho}{\sqrt{r}} \sigma_r (\mt{X})_1 + 2 D_\rho \tau \eta,
\label{eq:nuclear_norm_reconstruction}
\end{equation}
provided that $\mathcal{A}$ obeys a $r/\ell_q$-NSP.

Now, set $d=2^n$ and fix $\rho = \rho_0 \in (0,1)$, as well as
\begin{equation*}
C_1 \geq \frac{2 \tilde{C}_1}{\rho_0^6} \quad \textrm{and} \quad \gamma = 2 \rho_0^4 \tilde{\gamma},
\end{equation*}
where $\tilde{C}_1, \tilde{\gamma}_1$ are constants featuring in Theorem~\ref{thm:clifford_nsp}.
This statement then assures that choosing 
\begin{equation*}
m \geq C_1 \kappa (\vc{z}, r) r d \log (d) \geq \frac{ \tilde{C}_1}{\rho_0^6} \kappa (z,r) r d \log (2d)
\end{equation*}
random elements of a Clifford orbit $\Cli_n \cdot \vc{z} \vc{z}^*$ results in a measurement operator $\mathcal{A}: H_d \to \mathbb{R}^m$ that obeys a $r/\ell_q$-NSP with probability at least $1-\mathrm{e}^{-\frac{\gamma m}{\kappa (\vc{z}, r)}}$.
The associated constants are
\begin{equation*}
\rho= \rho_0\quad \textrm{and} \quad \tau = \frac{\tilde{C}_3}{\rho_0^2} \sqrt{ \kappa (z,r)} d m^{-\frac{1}{q}}.
\end{equation*}
Inserting these constants into \eqref{eq:nuclear_norm_reconstruction}
then implies
\begin{equation*}
\| \mt{Z}^\sharp - \mt{X} \|_2 \leq  \frac{C_2}{\sqrt{r}} \sigma_r (\mt{X}) + C_3 \sqrt{ \kappa (\vc{z},r)} d m^{-\frac{1}{q}} \eta
\end{equation*}
with constants
$
C_2 = \frac{2(1+\rho_0)^2}{1-\rho_0}$ and
$
C_3~=~\tilde{C}_3 \frac{(3+\rho_0)}{(1-\rho_0) \rho_0^2}.
$

\subsection{Extension to positive semidefinite matrix reconstruction} \label{sub:psd}

Suppose that the matrices of interest $\mt{X} \in H_d$ are not only approximately low rank, but also positive semidefinite. 
Also, we shall assume a phaseless measurement process $A_k = \vc{a}_k \vc{a}^*_k$ that is isotropic in the sense that
\begin{equation*}
\mathbb{E} \left[ A_k \right] = \frac{1}{d} \mathbb{I}.
\end{equation*}
Note that this is the case for random Clifford orbit measurements, since they obey equation~\eqref{eq:infty_design} for $t=1$.
For any $\beta \in (0,1)$ this expected concentration together with unit normalization of the $a_k$'s implies that the measurement matrices contained in a concrete sampling operator $\mathcal{A}:H_d \to \mathbb{R}^m$ obey
\begin{equation}
\mathrm{Pr} \left[ \left\| \frac{d}{m} \sum_{k=1}^m a_k a_k^* - \mathbb{I} \right\|_\infty \geq \beta \right] \leq d \mathrm{e}^{-\frac{4 \beta^2 m}{8(d+1)}}. 
\label{eq:Y_deviation}
\end{equation}
This follows from standard matrix concentration inequalities, see for instance \cite[Proof of Proposition 8.2]{kabanava_stable_2015}.
We shall fix $Y_{\mathcal{A}} := \frac{d}{m} \sum_{k=1}^m a_k a_k^* \in H_d$ and $\beta_0 = \frac{\sqrt{2}-1}{\sqrt{2}+1}$---which both may not be optimal---to simplify presentation in the remainder of this section.

Positive semi-definiteness of $X$ and $\| Y_{\mathcal{A}} - \mathbb{I} \|_{\infty} \leq \beta$ allows for replacing constrained nuclear norm minimization \eqref{eq:phaselift} by the simpler reconstruction algorithm \eqref{eq:nnls}:
\begin{equation*} 
	  \mt{Z}^\sharp = \underset{\mt{Z} \geq \mt{0}}{\textrm{argmin}} \left\| \mathcal{A}(\mt{Z}) - \vc{y} \right\|_{\ell_q}.
\end{equation*}

\begin{theorem}[Special case of Theorem 8.1 in \cite{kabanava_stable_2015}] \label{thm:nsp_psd}
Fix $r,q \geq 1$ and suppose that $\mathcal{A}: H_d \to \mathbb{R}^m$ obeys a $r/\ell_q$-NSP with parameters $\rho \in (0,\frac{1}{2})$ and $\tau>0$, as well as $\left\| Y_{\mathcal{A}} - \mathbb{I} \right\|_{\infty} \leq \beta_0$
Then, any pair $\mt{X}, \mt{Z} \in H_d$ of positive semidefinite matrices obeys
\begin{equation*}
\| \mt{Z} - \mt{X} \|_2
\leq \frac{ \bar{C}_\rho}{\sqrt{r}}  \sigma_r (\mt{X}) + \bar{D}_\rho \left( dm^{-\frac{1}{q}} + \tau \right) \| \mathcal{A} (\mt{X}-\mt{Z}) \|_{\ell_q}
\end{equation*}
with
$
\bar{C}_\rho = 4 \frac{(1+2\rho)^2}{1-2\rho}
$
and
$
\bar{D}_\rho = 2 \frac{3+2 \rho}{1-2 \rho}
$.  
\end{theorem}

Now note that the minimizer $Z^\sharp$ of \eqref{eq:nnls}, as well as any matrix $X$ of interest are positive semidefinite. Moreover, $Z^\sharp$ obeys
\begin{align*}
\left\| \mathcal{A}\left(\mt{X}-\mt{Z}^\sharp \right) \right\|_{\ell_q}
=& \left\| \vc{y}-\error - \mathcal{A}\left(\mt{Z}^\sharp \right) \right\|_{\ell_q} \\
\leq& \| \error \|_{\ell_q} + \left\| \mathcal{A} \left( \mt{Z}^\sharp \right) - \vc{y} \right\|_{\ell_q} \\
\leq &\| \error \|_{\ell_q} + \left\| \mathcal{A}(\mt{X}) - \vc{y} \right\|_{\ell_q} = 2 \| \error \|_{\ell_q},
\end{align*}
where the last inequality is due to the fact that $\mt{X}$ itself is a feasible point of the optimization \eqref{eq:nnls}. 
Inserting this bound into the assertion of Theorem~\ref{thm:nsp_psd} gives
\begin{equation}
\left\| \mt{Z}^\sharp - \mt{X} \right\|_2
\leq  \frac{ \tilde{C}_\rho}{\sqrt{r}} \sigma_r (\mt{X}) + 2 \tilde{D}_\rho (d m^{-\frac{1}{q}} + \tau ) \| \error \|_{\ell_q}.\label{eq:nnls_reconstruction}
\end{equation}

In order to derive Theorem~\ref{thm:main_psd}, we fix $d=2^n$,  $\rho = \rho_0 \in \left(0, \frac{1}{2} \right)$ and once more set $C_1 \geq \frac{2 \tilde{C}_1}{\rho_0^6}$ and $\gamma = 2 \rho_0^4 \tilde{\gamma}$. 
Then for any $r,q$,  Theorem~\ref{thm:clifford_nsp} assures that a measurement operator $\mathcal{A}: H_d \to \mathbb{R}^m$ containing $m \geq C_1 \kappa (z,r)r d \log (d)$ random elements of a Clifford orbit $\Cli_n \cdot \vc{z} \vc{z}^*$ obeys the $r/\ell_q$-NSP with parameters $\rho = \rho_0$ and $\tau = \frac{\tilde{C}_3}{\rho_0^2} \sqrt{ \kappa (z,r)} d m^{-\frac{1}{q}}$. The probability of failure for this to be true is bounded by $\mathrm{e}^{-\frac{\gamma m}{\kappa (z,r)}}$.

Moreover, Eq.~\eqref{eq:Y_deviation} asserts that the second condition in Theorem~\ref{thm:nsp_psd}, namely $\| Y_{\mathcal{A}} - \mathbb{I}\|_\infty \leq \beta_0$, is met with probability at least $1- d \mathrm{e}^{-\frac{\beta_0^2 m}{8(d+1)}}$. 

A union bound over these two individual probabilities of failure yields
\begin{align*}
\mathrm{e}^{-\frac{ \gamma m}{\kappa (\vc{z},r)}} + d \mathrm{e}^{-\frac{4 \beta_0^2 m}{8(d+1)}}
\leq & (d+1) \mathrm{e}^{-\frac{ \min \left\{ \frac{1}{2} \beta_0^2, \gamma \right\} m}{\max \left\{ \kappa (\vc{z},r), 8(d+1) \right\}}} \\
\leq & (d+1) \mathrm{e}^{-\frac{ \gamma m}{d+1}}
\end{align*}

Provided that both assertions hold true, Eq.~\eqref{eq:nnls_reconstruction} implies  
\begin{align*}
\left\| \mt{Z}^\sharp - \mt{X} \right\|_2
\leq &  \frac{\tilde{C}_{\rho_0}}{\sqrt{r}} \sigma_r(\mt{X}) + 2 \tilde{D}_{\rho_0} \left( d m^{-\frac{1}{q}} + \tau \right) \| \error \|_{\ell_q} \\
\leq &\frac{\hat{C}_2}{\sqrt{r}} \sigma_r (X) + \hat{C}_3 \sqrt{ \kappa (z,r)} d m^{-\frac{1}{q}} \| \error \|_{\ell_q}
\end{align*}
with constants $\hat{C}_2 = \frac{(1+2 \rho)^2}{1-2\rho}$ and $\hat{C}_3 = 4 \frac{3+2 \rho}{1-2 \rho} \left( 1 + \frac{\tilde{C}_3}{\rho_0^2} \right)$.

\subsection*{Acknowledgments}

This work has been supported by the Excellence Initiative of the German Federal and State Governments (Grant ZUK 81), the ARO under contract W911NF-14-1-0098 (Quantum Characterization, Verification, and Validation), and the DFG (SPP1798 CoSIP). 
Major parts of this project were undertaken while DG and RK participated in the \emph{Mathematics of Signal Processing} program of the Hausdorff Research Institute of Mathematics at the University of Bonn.

\section{Appendix: Reconstructing stabilizer states from stabilizer
measurements}
\label{sec:heuristic}

Here, we present a heuristic argument that suggests that $O(n^2)$ noise-free stabilizer measurements might be sufficient to identify an unknown stabilizer state. 
The argument neither suggests an algorithm, nor does it seem easy to base a rigorous proof on it.

We note that there are results (e.g.\ the presentation archived at \cite{aaronson_identifying_2006} and an announcement \cite{ashley_presentation} of results due to 
Montanaro,
Aaronson,
Chen, 
Gottesman, and Liew)
in quantum information stating that a stabilizer state can be identified from $O(n)$ measurements of stabilizer bases. 
These results come with matching converses, based on Holevo's bound.
Their $O(n)$ basis measurements involve $O(n)\, 2^n$ inner products -- exponentially more than we conjecture are necessary.
There is no direct contradiction, however, as the quantum model is weaker than the one employed here:
In the quantum setup, the squared inner products $|\langle a_k, x \rangle|^2$ need to be estimated through quantum mechanical experiments, while in the classical noise-free model, we have direct access to their values.

For the analysis, we assume that $x, a_1, \dots, a_m$ are uniformly drawn stabilizer states and set $y_k = |\langle x, a_k\rangle|^2$. 
Recovery is possible when
\begin{equation*}
	I(x:y_1, \dots, y_m|a_1, \dots, a_m) = H(x),
\end{equation*}
i.e.\ when the mutual information between the object $x$ to be recovered and the outcomes $y_k$, conditioned on the choices of measurement, reach the entropy of $x$.
Here, we compute $H(x)$ and 
\begin{equation*}
	I(x:y_k|a_k) = H(y_k|a_k).
\end{equation*}
using the results of Ref.~\cite{kueng_qubit_2015}.

Adopting the notation of Ref.~\cite{kueng_qubit_2015}, the entropy of $x$ is
\begin{equation*}
	H(x) = \log_2 | \operatorname{Stabs}(2,n) | \simeq \frac12n(n+1)=\mathcal{O}(n^2).
\end{equation*}
The approximation becomes tight as $n\to\infty$, as can be checked numerically using the explicit formulas in \cite{aaronson2004improved,kueng_qubit_2015}.

To compute $H(y_k|a_k)$, we need to find the distribution of $y_k$, i.e.\ of the squared inner product between a fixed stabilizer state $a_k$ and a random one $x$.
Let $K$ be the intersection between the Lagrangian subspaces associated with $a_k$ and with $x$, respectively.
Then according to \cite{kueng_qubit_2015}, the squared inner product $y_k$ is equal to $2^{\dim K - n}$ if the respective phase functions agree on $K$; else it is equal to $0$.
The former event occurs with probability $2^{-\dim K}$.
Based on this, one can compute the distribution of $y_k$ conditioned on $a_k$:
\begin{eqnarray*}
	\operatorname{Pr}[ y_k = 0 | a_k] &=& \sum_{l=0}^n (1-2^{-l}), \\
	\operatorname{Pr}[ y_k = 2^{-l} | a_k ] &=& 2^{-l}\frac{\kappa(2,n,l)}{|\operatorname{Stabs}(2,n)|}\qquad(l=0,\dots,n),
\end{eqnarray*}
where, following \cite{kueng_qubit_2015}, $\kappa(2,n,l)$ denotes the number of Lagrangian subspaces intersecting a given Lagrangian suspace in $l$ dimensions.
One can plug these expressions 
into a computer algebra system to compute the entropy of the distribution.
It turns out to converge as $n\to\infty$ to
\begin{equation*}
	I(x:y_k|a_k) = H(y_k|a_k) \simeq 1.719 = \mathcal{O}(1).
\end{equation*}

Now it is \emph{not} true that $I(x:y_1, \dots, y_m|a_1, \dots, a_m)$ equals $m\,I(x:y_k|a_k)$. 
But it seems to us to be a plausible assumption that the conditional mutual information increases roughly linearly with $m$ until it reaches $H(x)$.
If true, this would imply that recovery is possible from
\begin{equation*}
	m \simeq H(x)/I(x:y_k|a_k) = \mathcal{O}(n^2)
\end{equation*}
stabilizer measurements.

Verifying or disproving this statement is an interesting open problem.

\bibliographystyle{ieeetr}
\bibliography{cliffinv}

\begin{thebibliography}{10}

\bibitem{walther_question_1963}
A.~Walther, ``The question of phase retrieval in optics,'' {\em J. Mod.
  Optic.}, vol.~10, no.~1, pp.~41--49, 1963.

\bibitem{millane_phase_1990}
R.~Millane, ``Phase retrieval in crystallography and optics,'' {\em JOSA A},
  vol.~7, pp.~394--411, 1990.

\bibitem{fienup_phase_1987}
C.~Fienup and J.~Dainty, ``Phase retrieval and image reconstruction for
  astronomy,'' in {\em Image Recovery: Theory and Application}, pp.~231--275,
  Elsevier, 1987.

\bibitem{fienup_hubble_1993}
J.~R. Fienup, J.~C. Marron, T.~J. Schulz, and J.~H. Seldin, ``Hubble space
  telescope characterized by using phase-retrieval algorithms,'' {\em Appl.
  Optics}, vol.~32, no.~10, pp.~1747--1767, 1993.

\bibitem{chapman_femtosecond_2006}
H.~N. Chapman, A.~Barty, M.~J. Bogan, S.~Boutet, M.~Frank, S.~P. Hau-Riege,
  S.~Marchesini, B.~W. Woods, S.~Bajt, W.~H. Benner, {\em et~al.},
  ``Femtosecond diffractive imaging with a soft-{X}-ray free-electron laser,''
  {\em Nat. Phys.}, vol.~2, no.~12, pp.~839--843, 2006.

\bibitem{pfeiffer_phase_2006}
F.~Pfeiffer, T.~Weitkamp, O.~Bunk, and C.~David, ``Phase retrieval and
  differential phase-contrast imaging with low-brilliance {X}-ray sources,''
  {\em Nat. Phys.}, vol.~2, no.~4, pp.~258--261, 2006.

\bibitem{flammia_minimal_2005}
S.~T. Flammia, A.~Silberfarb, and C.~M. Caves, ``Minimal informationally
  complete measurements for pure states,'' {\em Found. Phys.}, vol.~35, no.~12,
  pp.~1985--2006, 2005.

\bibitem{peres_quantum_2006}
A.~Peres, {\em Quantum theory: concepts and methods}, vol.~57.
\newblock Springer, 2006.

\bibitem{heinosaari_quantum_2013}
T.~Heinosaari, L.~Mazzarella, and M.~M. Wolf, ``Quantum tomography under prior
  information,'' {\em Commun. Math. Phys.}, vol.~318, no.~2, pp.~355--374,
  2013.

\bibitem{baldwin_strictly_2016}
C.~H. Baldwin, I.~H. Deutsch, and A.~Kalev, ``Strictly-complete measurements
  for bounded-rank quantum-state tomography,'' {\em Phys. Rev. A}, vol.~93,
  no.~5, p.~052105, 2016.

\bibitem{carmeli_efficient_2016}
C.~Carmeli, T.~Heinosaari, M.~Kech, J.~Schultz, and A.~Toigo, ``Efficient pure
  state quantum tomography from five orthonormal bases,'' {\em preprint
  arXiv:1604.02970}, 2016.

\bibitem{fienup_phase_1982}
J.~R. Fienup, ``Phase retrieval algorithms: a comparison,'' {\em Appl. Optics},
  vol.~21, no.~15, pp.~2758--2769, 1982.

\bibitem{candes_phase_2013}
E.~J. Cand{\`e}s, Y.~C. Eldar, T.~Strohmer, and V.~Voroninski, ``Phase
  retrieval via matrix completion,'' {\em SIAM J. Imaging Sci.}, vol.~6,
  pp.~199--225, 2013.

\bibitem{candes_exact_2013}
E.~{C}and{\`e}s, T.~{S}trohmer, and V.~{V}oroninski, ``{P}hase{L}ift: {E}xact
  and stable signal recovery from magnitude measurements via convex
  programming,'' {\em {C}omm. {P}ure {A}ppl. {M}ath.}, vol.~66, pp.~1241--1274,
  2013.

\bibitem{candes_solving_2013}
E.~J. Cand{\`e}s and X.~Li, ``Solving quadratic equations via {PhaseLift} when
  there are about as many equations as unknowns,'' {\em Found. Comput. Math.},
  pp.~1--10, 2013.

\bibitem{gross_partial_2014}
D.~Gross, F.~Krahmer, and R.~Kueng, ``A partial derandomization of
  {P}hase{L}ift using spherical designs,'' {\em J. Fourier Anal. Appl.},
  pp.~1--38, 2014.

\bibitem{alexeev_phase_2014}
B.~Alexeev, A.~S. Bandeira, M.~Fickus, and D.~G. Mixon, ``Phase retrieval with
  polarization,'' {\em SIAM J. Imaging Sci.}, vol.~7, no.~1, pp.~35--66, 2014.

\bibitem{candes_coded_2015}
E.~J. Cand{\`e}s, X.~Li, and M.~Soltanolkotabi, ``Phase retrieval from coded
  diffraction patterns,'' {\em Appl. Comput. Harmonic Anal.}, vol.~39, no.~2,
  pp.~277--299, 2015.

\bibitem{gross_improved_2015}
D.~Gross, F.~Krahmer, and R.~Kueng, ``Improved recovery guarantees for phase
  retrieval from coded diffraction patterns,'' {\em Appl. Comput. Harmonic
  Anal.}, 2015.
\newblock doi:10.1016/j.acha.2015.05.004.

\bibitem{kueng_orthonormal_2015}
R.~Kueng, ``Low rank matrix recovery from few orthonormal basis measurements,''
  in {\em 2015 International Conference on Sampling Theory and Applications
  (SampTA)}, pp.~402--406, IEEE, 2015.

\bibitem{kueng_low_2015}
R.~{K}ueng, H.~{R}auhut, and U.~Terstiege, ``{L}ow rank matrix recovery from
  rank one measurements,'' {\em Appl. Comput. Harmonic Anal.}, 2015.
\newblock DOI:10.1016/j.acha.2015.07.007.

\bibitem{salanevich_polarization_2015}
P.~Salanevich and G.~E. Pfander, ``Polarization based phase retrieval for
  time-frequency structured measurements,'' in {\em 2015 International
  Conference on Sampling Theory and Applications (SampTA)}, pp.~187--191, IEEE,
  2015.

\bibitem{tropp_convex_2015}
J.~A. {T}ropp, ``{C}onvex recovery of a structured signal from independent
  random linear measurements,'' in {\em Sampling Theory, a Renaissance},
  pp.~67--101, Birkh{\"a}user/Springer, 2015.

\bibitem{krahmer_phase_2016}
F.~Krahmer and Y.-K. Liu, ``Phase retrieval without small-ball probability
  assumptions,'' {\em preprint arXiv:1604.07281}, 2016.

\bibitem{kabanava_stable_2015}
M.~Kabanava, R.~Kueng, H.~Rauhut, and U.~Terstiege, ``Stable low-rank matrix
  recovery via null space properties,'' {\em Inf. Inference}, 2016.
\newblock doi:10.1093/imaiai/iaw014.

\bibitem{kech_explicit_2015}
M.~Kech, ``Explicit frames for deterministic phase retrieval via phaselift,''
  {\em preprint arXiv:1508.00522}, 2015.

\bibitem{bodmann_algorithms_2016}
B.~G. Bodmann and N.~Hammen, ``Algorithms and error bounds for noisy phase
  retrieval with low-redundancy frames,'' {\em Appl. Comput. Harmonic Anal.},
  2016.

\bibitem{pohl_phaseless_2014}
V.~Pohl, F.~Yang, and H.~Boche, ``Phaseless signal recovery in infinite
  dimensional spaces using structured modulations,'' {\em J. Fourier Anal.
  Appl.}, vol.~20, no.~6, pp.~1212--1233, 2014.

\bibitem{delsarte_spherical_1977}
P.~Delsarte, J.-M. Goethals, and J.~J. Seidel, ``Spherical codes and designs,''
  {\em Geometriae Dedicata}, vol.~6, no.~3, pp.~363--388, 1977.

\bibitem{renes_symmetric_2004}
J.~M. Renes, R.~Blume-Kohout, A.~J. Scott, and C.~M. Caves, ``Symmetric
  informationally complete quantum measurements,'' {\em J. Math. Phys.},
  vol.~45, no.~6, pp.~2171--2180, 2004.

\bibitem{kueng_spherical_2015}
R.~Kueng, D.~Gross, and F.~Krahmer, ``Spherical designs as a tool for
  derandomization: The case of {PhaseLift},'' in {\em 2015 International
  Conference on Sampling Theory and Applications (SampTA)}, pp.~192--196, May
  2015.

\bibitem{ehler_phase_2015}
M.~Ehler, M.~Gr{\"a}f, and F.~J. Kir{\'a}ly, ``Phase retrieval using random
  cubatures and fusion frames of positive semidefinite matrices,'' {\em Waves,
  Wavelets and Fractals}, vol.~1, no.~1, 2015.

\bibitem{ambainis_quantum_2007}
A.~Ambainis and J.~Emerson, ``Quantum $t$-designs: $t$-wise independence in the
  quantum world,'' in {\em 22nd Annual IEEE Conference on Computational
  Complexity (CCC'07)}, pp.~129--140, June 2007.

\bibitem{brandao_local_2012}
F.~G. Brandao, A.~W. Harrow, and M.~Horodecki, ``Local random quantum circuits
  are approximate polynomial-designs,'' {\em preprint arXiv:1208.0692}, 2012.

\bibitem{zhu_multiqubit_2015}
H.~Zhu, ``Multiqubit {Clifford} groups are unitary 3-designs,'' {\em preprint
  arXiv:1510.02619}, 2015.

\bibitem{webb_clifford_2015}
Z.~Webb, ``The {Clifford} group forms a unitary 3-design,'' {\em preprint
  arXiv:1510.02769}, 2015.

\bibitem{kueng_qubit_2015}
R.~Kueng and D.~Gross, ``Qubit stabilizer states are complex projective
  3-designs,'' {\em preprint arXiv:1510.02767}, 2015.

\bibitem{klappenecker_mutually_2005}
A.~Klappenecker and M.~Rotteler, ``{Mutually unbiased bases are complex
  projective 2-designs},'' in {\em {2005 IEEE International Symposium on
  Information Theory (ISIT), Vols 1 and 2}}, pp.~{1740--1744}, {2005}.

\bibitem{bodmann_achieving_2015}
B.~G. Bodmann and J.~Haas, ``Achieving the orthoplex bound and constructing
  weighted complex projective 2-designs with {S}inger sets,'' {\em preprint
  arXiv:1509.05333}, 2015.

\bibitem{povm_paper}
R.~Kueng, H.~Zhu, M.~Grassl, and D.~Gross, ``Distinguishing quantum states
  using {Clifford} orbits,'' {\em preprint arXiv:1609.08595}, 2016.

\bibitem{perina_quantum_2012}
J.~Perina, Z.~Hradil, and B.~Jurco, {\em Quantum optics and fundamentals of
  physics}.
\newblock Springer, 2012.

\bibitem{banaszek_focus_2013}
K.~Banaszek, M.~Cramer, and D.~Gross, ``Focus on quantum tomography,'' {\em New
  J. Phys.}, vol.~15, no.~12, p.~125020, 2013.

\bibitem{foucart_mathematical_2013}
S.~Foucart and H.~Rauhut, {\em A {M}athematical {I}ntroduction to {C}ompressive
  {S}ensing}.
\newblock Applied and Numerical Harmonic Analysis, Birkh\"auser/Springer, New
  York, 2013.

\bibitem{donoho_compressed_2006}
D.~L. Donoho, ``Compressed sensing,'' {\em {IEEE} Trans. Inform. Theory},
  vol.~52, no.~4, pp.~1289--1306, 2006.

\bibitem{candes_near_2006}
E.~J. Cand{\`e}s and T.~Tao, ``Near-optimal signal recovery from random
  projections: Universal encoding strategies?,'' {\em {IEEE} Trans. Inform.
  Theory}, vol.~52, no.~12, pp.~5406--5425, 2006.

\bibitem{candes_stable_2006}
E.~J. Cand{\`e}s, J.~K. Romberg, and T.~Tao, ``Stable signal recovery from
  incomplete and inaccurate measurements,'' {\em Commun. Pur. Appl. Math.},
  vol.~59, no.~8, pp.~1207--1223, 2006.

\bibitem{candes_robust_2006}
E.~J. Cand{\`e}s, J.~Romberg, and T.~Tao, ``Robust uncertainty principles:
  Exact signal reconstruction from highly incomplete frequency information,''
  {\em {IEEE} Trans. Inform. Theory}, vol.~52, no.~2, pp.~489--509, 2006.

\bibitem{tao2003uncertainty}
T.~Tao, ``An uncertainty principle for cyclic groups of prime order,'' {\em
  preprint math/0308286}, 2003.

\bibitem{recht_guaranteed_2010}
B.~Recht, M.~Fazel, and P.~A. Parrilo, ``Guaranteed minimum-rank solutions of
  linear matrix equations via nuclear norm minimization,'' {\em SIAM Rev.},
  vol.~52, no.~3, pp.~471--501, 2010.

\bibitem{candes_exact_2009}
E.~J. Cand{\`e}s and B.~Recht, ``Exact matrix completion via convex
  optimization,'' {\em Found. Comput. Math.}, vol.~9, no.~6, pp.~717--772,
  2009.

\bibitem{candes_power_2010}
E.~J. Cand{\`e}s and T.~Tao, ``The power of convex relaxation: Near-optimal
  matrix completion,'' {\em {IEEE} Trans. Inform. Theory}, vol.~56, no.~5,
  pp.~2053--2080, 2010.

\bibitem{keshavan_matrix_2010}
R.~H. Keshavan, A.~Montanari, and S.~Oh, ``Matrix completion from a few
  entries,'' {\em {IEEE} Trans. Inform. Theory}, vol.~56, no.~6,
  pp.~2980--2998, 2010.

\bibitem{gross_quantum_2010}
D.~Gross, Y.-K. Liu, S.~T. Flammia, S.~Becker, and J.~Eisert, ``Quantum state
  tomography via compressed sensing,'' {\em Phys. Rev. Lett.}, vol.~105,
  no.~15, p.~150401, 2010.

\bibitem{gross_recovering_2011}
D.~Gross, ``{Recovering low-rank matrices from few coefficients in any
  basis},'' {\em {IEEE} Trans. Inform. Theory}, vol.~{57}, pp.~{1548--1566},
  {2011}.

\bibitem{liu_universal_2011}
Y.~K. Liu, ``Universal low-rank matrix recovery from pauli measurements,'' in
  {\em Advances in Neural Information Processing Systems 24}, pp.~1638--1646,
  Curran Associates, Inc., 2011.

\bibitem{FlamGLE12}
S.~T. Flammia, D.~Gross, Y.-K. Liu, and J.~Eisert, ``Quantum tomography via
  compressed sensing: error bounds, sample complexity and efficient
  estimators,'' {\em New J. Phys.}, vol.~14, no.~9, p.~095022, 2012.

\bibitem{gottesman_stabilizer_1997}
D.~Gottesman, {\em {‘Stabilizer codes and quantum error correction}}.
\newblock PhD thesis, California Institute of Technology, Pasadena, CA, 1997.

\bibitem{macwilliams1977theory}
F.~MacWilliams and N.~Sloane, {\em The Theory of Error-correcting Codes}.
\newblock North-Holland mathematical library, North-Holland Publishing Company,
  1977.

\bibitem{nebe_self_2006}
G.~Nebe, E.~M. Rains, and N.~J.~A. Sloane, {\em Self-dual codes and invariant
  theory}, vol.~17.
\newblock Springer, 2006.

\bibitem{grochenig2013foundations}
K.~Gr{\"o}chenig, {\em Foundations of Time-Frequency Analysis}.
\newblock Applied and Numerical Harmonic Analysis, Birkh{\"a}user Boston, 2013.

\bibitem{pfander-chapter}
G.~Pfander, ``Gabor frames in finite dimensions,'' in {\em Finite Frames}
  (P.~G. Casazza and G.~Kutyniok, eds.), Applied and Numerical Harmonic
  Analysis, pp.~193--239, 2013.

\bibitem{v1931eindeutigkeit}
J.~v.~Neumann, ``Die {E}indeutigkeit der {S}chr{\"o}dingerschen {O}peratoren,''
  {\em Math. Ann.}, vol.~104, no.~1, pp.~570--578, 1931.

\bibitem{mackey1955theory}
G.~Mackey, {\em The Theory of Unitary Group Representations}.
\newblock Chicago lectures in mathematics, University of Chicago Press, 1955.

\bibitem{Foll89book}
G.~B. Folland, {\em Harmonic Analysis in Phase Space}, vol.~122 of {\em Annals
  of Mathematics Studies}.
\newblock Princeton, NJ: Princeton University Press, 1989.

\bibitem{nielsen_quantum_2010}
M.~A. Nielsen and I.~L. Chuang, {\em {Quantum Computation and Quantum
  Information 10th Anniversary Edition}}.
\newblock Cambridge University Press, 2010.

\bibitem{gross_hudson_2006}
D.~Gross, ``Hudson's theorem for finite-dimensional quantum systems,'' {\em J.
  Math. Phys.}, vol.~47, no.~12, p.~122107, 2006.

\bibitem{hostens2005stabilizer}
E.~Hostens, J.~Dehaene, and B.~De~Moor, ``Stabilizer states and {Clifford}
  operations for systems of arbitrary dimensions and modular arithmetic,'' {\em
  Phys. Rev. A}, vol.~71, no.~4, p.~042315, 2005.

\bibitem{gross2008lu}
D.~Gross and M.~Van~den Nest, ``The {LU-LC} conjecture, diagonal local
  operations and quadratic forms over {GF (2)},'' {\em Quantum Inform.
  Comput.}, vol.~8, no.~3, pp.~263--281, 2008.

\bibitem{other_paper}
H.~Zhu, R.~Kueng, M.~Grassl, and D.~Gross, ``The {Clifford} group fails
  gracefully to be a unitary 4-design,'' {\em preprint arXiv:1609.08172}, 2016.

\bibitem{walls_quantum_2008}
D.~F. Walls and G.~J. Milburn, {\em Quantum Optics}.
\newblock Springer, 2nd~ed., 2008.

\bibitem{hudson_wigner_1974}
R.~L. Hudson, ``When is the {W}igner quasi-probability density non-negative?,''
  {\em Rep. Math. Phys.}, vol.~6, no.~2, pp.~249--252, 1974.

\bibitem{debeaudrap_linearized_2011}
N.~De~Beaudrap, ``A linearized stabilizer formalism for systems of finite
  dimension,'' {\em preprint arXiv:1102.3354}, 2011.

\bibitem{appleby_symmetric_2005}
D.~M. Appleby, ``Symmetric informationally complete--positive operator valued
  measures and the extended {Clifford} group,'' {\em J. Math. Phys.}, vol.~46,
  no.~5, p.~052107, 2005.

\bibitem{helsen_representations_2016}
J.~Helsen, J.~J. Wallman, and S.~Wehner, ``Representations of the multi-qubit
  {Clifford} group,'' {\em preprint arXiv:1609.08188}, 2016.

\bibitem{Runge96}
B.~Runge, ``Codes and {S}iegel modular forms,'' {\em Discrete Math.}, vol.~148,
  no.~1, pp.~175 -- 204, 1996.

\bibitem{balan_painless_2009}
R.~{Balan}, B.~G. {Bodmann}, P.~G. {Casazza}, and D.~{Edidin}, ``{Painless
  reconstruction from magnitudes of frame coefficients.},'' {\em {J. Fourier
  Anal. Appl.}}, vol.~15, pp.~488--501, 2009.

\bibitem{demanet_stable_2014}
L.~Demanet and P.~Hand, ``Stable optimizationless recovery from phaseless
  linear measurements,'' {\em J. Fourier Anal. Appl.}, vol.~20, no.~1,
  pp.~199--221, 2014.

\bibitem{shenoy_weyl_1994}
R.~G. Shenoy and T.~W. Parks, ``The weyl correspondence and time-frequency
  analysis,'' {\em {IEEE} Trans. Signal Process.}, vol.~42, no.~2,
  pp.~318--331, 1994.

\bibitem{kozek_spectral_1996}
W.~Kozek, ``Spectral estimation in non-stationary environments,'' {\em PhD
  Thesis}, 1996.

\bibitem{koltchinskii_bounding_2015}
V.~Koltchinskii and S.~Mendelson, ``Bounding the smallest singular value of a
  random matrix without concentration,'' {\em Internat. Math. Res. Notices},
  pp.~12991--13008, 2015.

\bibitem{mendelson_learning_2015}
S.~{M}endelson, ``{L}earning without {C}oncentration,'' {\em {J}. {A}{C}{M}},
  vol.~62, no.~3, pp.~1--25, 2015.

\bibitem{bannai_survey_2009}
E.~Bannai and E.~Bannai, ``A survey on spherical designs and algebraic
  combinatorics on spheres,'' {\em European J. Combin.}, vol.~30, no.~6,
  pp.~1392--1425, 2009.

\bibitem{mohan_iterative_2010}
K.~Mohan and M.~Fazel, ``Iterative reweighted least squares for matrix rank
  minimization,'' in {\em Proceedings of the Allerton Conference},
  pp.~653--661, 2010.

\bibitem{recht_null_11}
B.~Recht, W.~Xu, and B.~Hassibi, ``Null space conditions and thresholds for
  rank minimization,'' {\em Math. Program.}, vol.~Ser B, 127, pp.~175--211,
  2011.

\bibitem{recht_necessary_2008}
B.~Recht, W.~Xu, and B.~Hassibi, ``Necessary and sufficient conditions for
  success of the nuclear norm heuristic for rank minimization,'' in {\em Proc.
  47th IEEE Conference on Decision and Control}, pp.~3065--3070, 2008.

\bibitem{fornasier_low_2011}
M.~{F}ornasier, H.~{R}auhut, and R.~{W}ard, ``{L}ow-rank matrix recovery via
  iteratively reweighted least squares minimization,'' {\em {S}{I}{A}{M} {J}.
  {O}ptim.}, vol.~21, no.~4, pp.~1614--1640, 2011.

\bibitem{aaronson_identifying_2006}
S.~Aaronson and D.~Gottesman, ``Identifying stabilizer states,'' 2006.
\newblock \url{http://pirsa.org/08080052/}.

\bibitem{ashley_presentation}
A.~Montanaro, ``Three quantum learning algoirthms,'' 2013.
\newblock
  \url{https://people.maths.bris.ac.uk/~csxam/presentations/learningtalk.pdf}.

\bibitem{aaronson2004improved}
S.~Aaronson and D.~Gottesman, ``Improved simulation of stabilizer circuits,''
  {\em Phys. Rev. A}, vol.~70, no.~5, p.~052328, 2004.

\end{thebibliography}

\end{document}